\documentclass[journal,twoside,web]{ieeecolor}
\usepackage{generic}
\usepackage{cite}
\usepackage{amsmath,amssymb,amsfonts}
\usepackage{algorithmic}
\usepackage{graphicx}
\usepackage{textcomp}
\def\BibTeX{{\rm B\kern-.05em{\sc i\kern-.025em b}\kern-.08em
    T\kern-.1667em\lower.7ex\hbox{E}\kern-.125emX}}
\usepackage{cite}
\usepackage{amsmath,amssymb,amsfonts,verbatim}
\usepackage{algorithmic}
\usepackage{graphicx}
\usepackage[utf8]{inputenc}
\usepackage{textcomp}
\usepackage{xcolor}
\usepackage{braket}
\usepackage{soul}
\usepackage[colorinlistoftodos]{todonotes}
\newtheorem{theorem}{Theorem}
\newtheorem{lemma}{Lemma}

\newtheorem{definition}{Definition}

\DeclareMathOperator{\Tr}{Tr}
\usepackage{tikz}
\usetikzlibrary{decorations.pathmorphing} 
\usetikzlibrary{matrix} 
\usetikzlibrary{arrows} 
\usetikzlibrary{calc} 
\usetikzlibrary{arrows.meta}
\usetikzlibrary{positioning}
\usetikzlibrary{arrows}
\usetikzlibrary{calc}
\tikzset{dexteritas/.cd,
shifted path label/.style={pos=0.5,draw=none,rectangle,auto,sloped}}
\tikzset{shifted path/.style args={from #1 to #2 by #3}{insert path={
let \p1=($(#1.east)-(#1.center)$),
\p2=($(#2.east)-(#2.center)$),\p3=($(#1.center)-(#2.center)$),
\n1={veclen(\x1,\y1)},\n2={veclen(\x2,\y2)},\n3={atan2(\y3,\x3)} in
(#1.{\n3+180+asin(#3/\n1)}) to 
(#2.{\n3-asin(#3/\n2)})
}}}
\tikzset{labeled shifted path/.style args={from #1 to #2 by #3 label #4}{insert path={
let \p1=($(#1.east)-(#1.center)$),
\p2=($(#2.east)-(#2.center)$),\p3=($(#1.center)-(#2.center)$),
\n1={veclen(\x1,\y1)},\n2={veclen(\x2,\y2)},\n3={atan2(\y3,\x3)} in
(#1.{\n3+180+asin(#3/\n1)}) to node[dexteritas/shifted path label]{#4}
(#2.{\n3-asin(#3/\n2)})
}}}
\tikzstyle{block} = [draw,align=center,rectangle,thick,minimum height=4em,minimum width=4em,text width=3cm]
\tikzstyle{state} = [draw,align=center, circle,inner sep=0mm,minimum size=1.7cm,text width=1cm]
\tikzstyle{connector} = [->,thick]
\tikzstyle{line} = [thick]
\tikzstyle{guide} = []

\newcommand{\EX}{\mathbb{E}}
\newcommand{\PR}{\mathbb{P}}
\newcommand{\BN}{\mathbb{N}}
\newcommand{\CX}{\mathcal{X}}
\newcommand{\CA}{\mathcal{A}}
\newcommand{\CS}{\mathcal{S}}
\newcommand{\CH}{\mathcal{H}}
\newcommand{\CE}{\mathcal{E}}
\newcommand{\CL}{\mathcal{L}}
\newcommand{\complex}{\mathbb{C}}
\newcommand{\utility}{\theta}
\newcommand{\control}{u_k}
\newcommand{\act}{a}
\newcommand{\ly}{Lyapunov }
\newcommand{\limset}{$D_{\infty}$}
\DeclareMathOperator{\vect}{vec}
\newcommand{\state}{s}
\newcommand{\belief}{\pi}
\newcommand{\posterior}{\alpha_k}
\newcommand{\joint}{\eta_k}
\newcommand{\DLsignal}{\beta_k}
\newcommand{\id}{\mathbb{I}}
\newcommand{\param}{\Phi}
\newcommand{\parama}{\phi_1}
\newcommand{\paraml}{\phi_2}
\newcommand{\paramp}{\phi_3}
\newcommand{\diag}{\text{diag}}
\newcommand{\reals}{{\rm I\hspace{-.07cm}R}}
\def\BibTeX{{\rm B\kern-.05em{\sc i\kern-.025em b}\kern-.08em
    T\kern-.1667em\lower.7ex\hbox{E}\kern-.125emX}}
\begin{document}
%
\title{\LARGE \bf
\ly based  Stochastic Stability of a Quantum Decision System for Human-Machine Interaction
}
%
%
%

\author{Luke~Snow,~\IEEEmembership{}
        Shashwat~Jain~\IEEEmembership{}
        and~Vikram Krishnamurthy,~\IEEEmembership{Fellow,~IEEE}
\thanks{Partial results of this paper were submitted to IEEE CDC 2022.}
\thanks{This research was supported in part by the National Science Foundation grant CCF-2112457 and Army Research Office grant W911NF-19-1-0365}
\thanks{Luke Snow {\tt\small las474@cornell.edu}, Shashwat Jain {\tt\small sj474@cornell.edu}, Vikram Krishnamurthy {\tt\small vikramk@cornell.edu} are  with the School of Electrical and Computer Engineering, Cornell University, Ithaca, NY 14853, USA
        }%
\thanks{All three  authors contributed equally to this work.}}
\maketitle

\begin{abstract}
In mathematical psychology,  decision makers are modeled using the Lindbladian equations from quantum mechanics to capture important human-centric features such as  order effects and violation of the sure thing principle. We consider   human-machine interaction involving a quantum decision maker (human)  and a controller (machine). Given a sequence of  human decisions over time,  how can the controller dynamically provide input messages to  adapt these decisions so as to  converge to a specific decision?  We show via  novel stochastic \ly arguments  how  the Lindbladian dynamics of the quantum decision maker can be controlled to converge to a specific decision asymptotically.  Our methodology yields a useful mathematical framework for human-sensor decision making. The stochastic \ly results are also of independent interest as they generalize recent results in the literature.
\end{abstract}


%
\IEEEpeerreviewmaketitle

\section{Introduction}

Recent studies in mathematical psychology \cite{martinez2016quantum}, \cite{kvam2021temporal}, \cite{busemeyer2020application}, show that the Lindbladian equations from quantum mechanics facilitate modeling peculiar aspects  of human decision making. Specifically, such quantum decision models capture {\em order effects}
(humans  perceive $P(H|A\cap B) $ and $P(H|B \cap A)$ differently in decision making)
, violation of the {\em sure thing principle} (human perception of probabilities in decision making violates the total probability rule), and temporal oscillations in decision preferences. Motivated by the design of human-machine interaction systems, this paper addresses the following question: {\em Given a sequence of human decisions over time, how can a controller (machine)  adapt the Lindbladian dynamics  (of the human decision maker) so as to converge to a specific decision?} To investigate this, we develop a novel generalization of recent results involving finite-step stochastic \ly functions. Thus at an abstract level, we study the stochastic stability of a switched controlled Lindbladian dynamic system where the switching occurs due to the interaction of the controller (machine)  and decision maker (human) at specific (possibly random) time instants. 

\subsection{Decision Making Context}
Figure~\ref{fig:Model} shows our schematic setup. The finite-valued random variable $\state\sim \belief_0(\cdot)$ 
denotes the underlying state of nature, where $\belief_0$ is a known probability mass function.
The input signals $y_k$ and $z_k$ are noisy observations of the state with conditional observation densities $p(y|\state)$ and $p(z|\state)$, respectively. The observation $y_k$ is fed into a bayesian recommendor which recommends humans the most probable underlying state $\posterior$, thereby influencing the the human's psychological state $\rho$. This state is represented as a density operator in Hilbert Space, which evolves via the Lindbladian equation parametrized by the recommendation $\posterior$ and control input $\control$. The controller gets its own version of the knowledge $\DLsignal$ of the underlying state from the sensor which observes the underlying state as $z_k$. The controller detects the actions $\act_k$ and using the sensor input $\DLsignal$ it outputs a control signal to $\control$ to the human.

The density operator $\rho$ encodes a probability distribution over actions $\{a_j\}_{j \in \{1,\dots,m\}}$, and at each time point an action is taken according to this distribution. The machine observes the actions and outputs a feedback control signal $\control$ to the human.


\subsubsection*{Examples} Several
examples in robotics \cite{askarpour2019formal},  interactive marketing/advertising \cite{belanche2020consumer}, recommender systems \cite{lu2012recommender} and control of game-theoretic economic interactions \cite{9280374} exploit models for human decision making.
One example is a machine assisted healthcare system for patients with dementia \cite{HOEY2010503}, in which the patient is assisted by a machine (smart watch) to wash his hands. The machine's sensor detects whether a certain set of sequential actions are followed by the patient, and then sends those results to a controller which gives an audio/video command to the patient. 
In this context the underlying states ($s$) are the tap water, soap dispenser and towel dispenser, which are partially observed by both sensor and the patient. The patient has a psychological state ($\rho$) and the resultant hand washing actions ($a_k$) are detected by the controller which gives the control input ($\control$). In our work, we model the psychological state of the patient as a Lindbladian evolution as shown in Figure~\ref{fig:Model} since this accounts for a wider range of human behavior, such as irrational decisions which could be made by the dementia patient, than classical models.

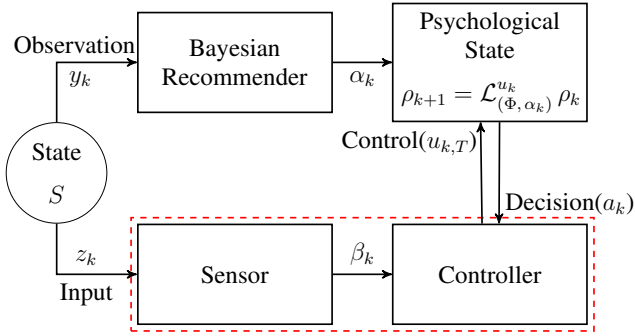
\begin{figure}[h!]
    \centering
    \resizebox{0.5\textwidth}{!}{%
  \begin{tikzpicture}[scale=1, auto, >=stealth']
	\large
	\matrix[ampersand replacement=\&, row sep=0mm, column sep=0.5cm] {
		
		\&\&\&\&\\
		\&\&\node[block] (F1) {Bayesian Recommender};\&\&\node[block] (S1) {Psychological State \[\rho_{k+1}=\mathcal{L}_{(\Phi,\,\posterior)}^{\control}\,\rho_k\]};\\
		\&\node[state] (X){State \[S\]};\&\&\&\\
		\&\&\node[block] (F2) {Sensor};\&\&\node[block] (f1) {Controller};
		\\
	};
	\draw [connector] (F2) --node[above]{$\DLsignal$} (f1);
	\draw[connector] (F1)--node[below]{$\posterior$}(S1);
	\draw [connector] ($(X.north) - (0cm, 0cm)$) -- ++(0,0cm)|- node [right=0.3cm,above]
	{Observation} node[right=0.4cm,below]{$y_k$} ($(F1.west) - (0cm, 0em)$);
	\draw [connector] ($(X.south) - (0cm, 0em)$)-- ++(0,0cm)|- node [right=0.5,below]
	{Input}node[right=0.5cm,above]{$z_k$} ($(F2.west) - (0cm, 0em)$); 
	\draw[dashed,thick, red] ($(F2.north west)+(-0.1,0.1)$) rectangle ($(f1.south east)+(0.1,-0.1)$);
	\draw[->,thick, shifted path=from f1 to S1 by 8pt]node[left=1.2cm, below]{Control($u_{k,T}$)};
    \draw[->, thick,shifted path=from S1 to f1 by 8pt]node[right=1.2,above]{Decision($\act_{k}$)};
\end{tikzpicture} 
  }
    \caption{Human-Machine Interaction Model. The state is observed by both Sensor and Bayesian Recommender which send a posterior to both Controller and Human. The Human and Machine interact via action and control feedback, where Human acts based upon upon by the control input given by the controller and the Bayesian Recommender. The machine controller uses input from sensor and actions of the human and gives a control feedback to human after $T$ time-steps. }
    \label{fig:Model}
\end{figure}
\subsection{Main Results and Organization}
    Given the described human-machine decision making system, the question we ask is: Can the human's decision preference be guided by the input control signals such that a desired target action is eventually taken at every time step? Our results reveal that this indeed is the case. We show this by developing a novel \ly stability result for a Lindbladian dynamic system.
    
    The main results and organization of the paper is as follows:
    \begin{enumerate}
        \item Sec.~\ref{sec2} introduces the open-quantum cognitive model of Martinez et. al. \cite{martinez2016quantum} and formulates a difference equation for the human psychological state. This time discretized form is amenable to the stochastic Lyapunov techniques presented in \cite{AMINI20132683}. Thus we form a bridge between the specific Lindbladian psychological structure of \cite{martinez2016quantum} and the more abstract quantum control formalisms of \cite{AMINI20132683}. However, in our formulation the machine interacts with the human decision maker over random time interval lengths $T$, and thus we need a more general stochastic Lyapunov argument, which we discuss in point (3) below.
        \item Sec.~\ref{sec3} presents Theorem~\ref{thm:convg}, which is our main result and shows the stochastic stability of our human-machine decision system. The proof uses the methodology of Amini et. al \cite{AMINI20132683}, along with \ly techniques and Theorem~\ref{thm:extension} which we provide in Sec. IV.
        \item Sec.~\ref{sec4} provides a generalization of a finite-step \ly stability result given in Qin et. al. \cite{qin2019lyapunov} in Theorem~\ref{thm:rand}, to the case when the finite-step interval $T$ is a random variable. Also Theorem~\ref{thm:extension} is a modified form of this result which is used to prove Theorem~\ref{thm:convg} in Sec.~\ref{sec3}.
    \end{enumerate}
    Thus we integrate a novel finite-step \ly stability result (developed in Sec. IV) with the quantum control formalism of \cite{AMINI20132683} to prove controlled convergence of a high-level quantum representation (in Sec. III). This directly corresponds (Sec. I) to the convergence of a specific Lindbladian structure which is a useful model of psychological preference evolution. The results presented develop a novel formulation of a machine controlled human decision maker.

\subsection{Literature Review}
Generative models for human decision making are studied extensively in behavioral economics and psychology.
 The classical  formalisms  of human decision making are the Expected Utility models of Von-Neumann and Morgenstern (1953)\cite{morgenstern1953theory} and Savage (1954) \cite{savage1951theory}. Despite the successes of these models, numerous experimental findings, most notably those of Kahneman and Tverksy \cite{kahneman1982judgment}, have demonstrated violations of the proposed decision making axioms. There have since been subsequent efforts to develop axiomatic systems which encompass wider ranges of human behavior, such as the Prospect Theory \cite{kahneman2013prospect}. However, given the complexity of human psychology and behavior it is no surprise that current models still have points of failure. The theory of Quantum Decision Making (\cite{busemeyer2012quantum}, \cite{khrennikov2010ubiquitous}, \cite{yukalov2010mathematical} and references therein) has emerged as a new paradigm which is capable of generalizing current models and accounting for certain violations of axiomatic assumptions. For example, it has been empirically shown that humans routinely violate Savage's 'Sure Thing Principle' \cite{khrennikov2009quantum}, \cite{aerts2011quantum}, which is equivalent to violation of the law of total probability, and that human decision making is affected by the order of presentation of information \cite{trueblood2011quantum} \cite{busemeyer2011quantum} ("order effects"). These violations are natural motivators for treating the decision making agent's mental state as a quantum state in Hilbert Space; The mathematics of quantum probability was developed as an explanation of observed self-interfering and non-commutative behaviors of physical systems, directly analogous to the findings which Quantum Decision Theory (QDT) aims to treat. 
 
 {\em Remark}. QDT  models in  psychology do not claim that the brain is acting as a quantum device in any physical sense. Instead QDT  serves as a {\em parsimonious generative blackbox model} for human decision making that is backed up by extensive experimental studies \cite{kvam2021temporal}, \cite{busemeyer2012quantum}.

Within Quantum Decision Theory, several recent advances have utilized quantum dynamical systems to model time-evolving decision preferences. The classical model for this type of time-evolving mental state is a Markovian model, but in \cite{busemeyer2009empirical} an alternative formulation based on Schr\"{o}dinger's Equation is developed. This model is shown to both reconcile observed violations of the law of total probability via quantum interference effects and model choice-induced preference changes via quantum projection. This is further advanced in \cite{asano2012quantum}, and  \cite{martinez2016quantum} where the mental state is modeled as an open-quantum system. This open-quantum system representation allows for a generalization of the widely used Markovian model of preference evolution, while maintaining these advantages of the quantum framework. Busemeyer et. al. \cite{kvam2021temporal} provide empirical analysis which supports the use of open-quantum models and conclude "An open system model that incorporates elements of both classical and quantum dynamics provides the 
best available single system account of these three characteristics—evolution, oscillation, and choice-induced 
preference change". 

To prove stability of this human-machine interaction, we generalize a recent result in finite-step \ly stability given in Qin et. al \cite{qin2019lyapunov}, which itself is a generalization of the work in \cite{aeyels1998new} from deterministic to stochastic systems.  \cite{qin2019lyapunov} aims to advance \ly theory for stochastic dynamical systems, in particular for network algorithms for distributed computation, and applies the developments to several distributed computation problems. They show how the results can be directly used to study problems in which one needs to prove convergence of inhomogeneous Markov chains, such as in distributed optimization \cite{nedic2014distributed} \cite{nedic2009distributed}, other distributed coordination algorithms \cite{tahbaz2009consensus} \cite{cao2008agreeing} \cite{tahbaz2008necessary}, and solving linear algebraic equations \cite{liu2017exponential} \cite{mou2015distributed}. Thus our finite-step stochastic \ly stability results can find application in these domains, but we leave this as future work.

\subsection*{Notation}
\begin{enumerate}
    \item $\ket{n} \in \mathcal{S}$ is $n^{th}$ basis vector in a Hilbert space $\mathcal{S}$
    \item $\mathcal{H}(\{\ket{\state_1}, \dots, \ket{\state_k}\})$: the Hilbert Space spanned by the orthonormal basis vector set $\{\ket {\state_1},\dots,\ket {\state_k}\}$
    \item $A^{\dagger}$: adjoint of $A$
    \item $\ket{n}\bra{m}$: outer product of $\ket{n}$ and $\ket{m}$
    \item Density operator $\rho$: $\rho=\ket{\Psi}\bra{\Psi}$ for some $\ket{\Psi}\in \CS$
    \item Trace of operator $A$: $\Tr(A)=\sum_{l=1}^{N}\bra{n}A\ket{n}$ for basis $\ket{n}\in\CS$
    \item  Random events are defined on $(\Omega,\mathcal{F},\PR)$.
    \item $\sigma(x_1,\dots,x_k)$: sigma algebra generated by $\{x_i\}_{i=1}^k$
\end{enumerate}
\section{Quantum Probability Model for Human Decision Making}
\label{sec2}
This  section presents the open-quantum system model that we will  use to represent the decision preference evolution of the human decision maker.
We define the evolution of the density operator of the decision maker using the open-system Quantum Lindbladian Equation, proposed in \cite{martinez2016quantum} and implemented in \cite{busemeyer2020application}. Reference~\cite{kvam2021temporal} provides empirical evidence which concludes that this open-system structure is the most parsimonious model which can capture observations of dynamical preference evolution such as oscillation and choice-induced preference change. We first illustrate the high-level psychological model in Subsection \ref{model}, the detailed mathematical and psychological derivation of this model is given in Appendix~\ref{constr_lind}. Sec. \ref{dt_deriv},  Sec.\ref{model_practicality} and Appendix \ref{Lindblad_param} present the derivation and construction of our Lindbladian psychological model and discuss why this model is useful. The reader may skip these sections and go directly to \ref{model} and then to Sec.~\ref{sec3} and Sec.~\ref{sec4} for the main results.

\subsection{Lindbladian Preference Evolution Model}
\label{model}


The psychological model consists of the following: 
\begin{enumerate}
\item  A state of nature $s \in \{1,\ldots,n\} = \CX$, with probability mass function $\pi_0(s)$
\label{2Afirst}
\item A state posterior $\posterior$ obtained by a Bayesian sensor from noisy observation $y_k \sim p(y_k | s)$ and prior $\alpha_{k-1}$. $\posterior$ parameterizes the human psychological evolution, as quantified in Appendix \ref{Lindblad_param}.
\item A state information signal $\DLsignal$ which depends on the noisy observation $z_k \sim p(z_k|s)$. For example, this could represent the output of a deep learning classifier. This is discussed further in Sec.~\ref{machine_control}.
\item  An action $a_k \in \{1,\ldots,m\} = \CA$ by the human at discrete time~$k$, for $k = 0,1,2,\ldots$. 
\item An objective state-action utility function 
\begin{equation}
\label{zeta}
    \zeta: \CX \times \CA \rightarrow \reals
\end{equation} 
which gives a real-valued utility to each unique state-action pair. We call this the \textit{objective} utility function because it will be augmented by several psychological parameters to produce a \textit{subjective} utility function eq. (\ref{subutil}), which will influence the human's action preference evolution, see Appendix \ref{Lindblad_param}.
\item  Scalar control input $\control \in [-\bar{u},\bar{u}],\, \bar{u} \in \reals_+$, generated by the machine every $T$ time steps, where $T$ is a random variable, $T: \Omega \rightarrow \mathbb{N}$. The control input $\control$ is piecewise constant for $T$ time-steps.  $T$ has a known probability mass function $\pi_T(\cdot)$. The input $\control$ controls the parameters of the psychological state evolution \eqref{M_update}, and thereby influences the human decision making. $\control$ models a recommendation signal (for example a posterior probability of the state of nature) by the machine to the human. These must satisfy the constraints in Sec. \ref{machine_control}.
\item The discrete time evolution of the psychological state
\begin{equation}
\label{M_update}
    \rho_{k+T} = \mathbb{M}_{a_k, T}^{\control}\,\rho_k := \frac{M_{a_k, T}^{\control}\, \rho_k\, M_{a_k, T}^{\control\dagger}}{\Tr(M_{a_k, T}^{\control}\, \rho_k\, M_{a_k, T}^{\control\dagger})}
\end{equation} 
where $k = 0,1,\ldots$ denotes discrete-time. 
  $a_k$ is a $T$-length sequence $\{a_{k_i}\}_{i = 1}^T $of random actions $a_{k_i}$ taking values in $\mathcal{A}$, and
$\rho_k$ is a complex-valued $d \times d$ matrix.
We derive $M_{a_k, T}^{\control}$ in the next subsection as a time dicretization of the Lindbladian ordinary differential equation from quantum mechanics.
\item The action probability distribution at  time $k$
\label{2Alast}
\begin{equation}
\label{eq:ProbEvo}
 \PR(a_k)=\Tr(M_{a_k}^{\control}\, \rho_k \, M_{a_k}^{\control\dagger}), a_k \in \{1,\cdots,m\}   
\end{equation}
In mathematical psychology, the distribution over actions \eqref{eq:ProbEvo} is
called the human's preference distribution~\cite{martinez2016quantum}. It can be observed that $\sum_{a_k}\Tr(M_{a_k}^{\control}\, \rho_k \, M_{a_k}^{\control\dagger})= 1$ and $\Tr(M_{a_k}^{\control}\, \rho_k \, M_{a_k}^{\control\dagger})\geq 0\ \forall a_k \in \{1,\dots,m\}$.
\end{enumerate}

Equation \eqref{M_update} represents the discrete-time psychological evolution of a human decision maker. It is derived from the time  discretization of the quantum mechanical Lindbladian differential equation, which quantifies the evolution of an open-quantum system; that is, a quantum system which interacts with a dissipative external environment. It has been shown in  \cite{busemeyer2020application} \cite{martinez2016quantum} that since this model incorporates aspects of both classical and quantum state evolution, it more accurately describes a wider range of psychological phenomena than classical or quantum models alone. 

 Note that the induced action choices are inherently probabilistic, with probabilities defined in \eqref{eq:ProbEvo}, rather than classically deterministic. The relevance of probabilistic decision making models, in particular quantum models, is discussed in \cite{pothos2009quantum}. The human-machine interactions also occur every $T$ time steps, with $T$ a random variable. This is useful because it guarantees convergence of sample paths even for randomly initialized human-machine interaction time scales. This could correspond to uncertainty in the decision/processing speed of a particular human (or between different humans) or of a particular machine. On a practical note, this model incorporates the randomness of human reaction times in decision making.
 
 We prove this convergence with random $T$ in the proof of Theorem~\ref{thm:convg}, and this uses our finite-step \ly result provided as Theorem~\ref{thm:extension}. An interesting technical note is that this proof cannot be extended to the case when $T$ is a nontrivial random process. Our proof induces process sub-sequences which each converge and exhaust the entire process sequence, therefore implying convergence of the sequence. If~$T$ is  time-evolving (e.g. a recurrent Markov chain) then the same argument will result in sub-sequences which do not exhaust the original sequence, and therefore overall convergence cannot be shown. However, $T$ \textit{can} be a random process on the positive integers which either evolves according to a Markov chain with an absorbing state or is non-increasing eventually with probability one. We do not provide proofs for these straightforward generalizations, they  may be useful in modeling evolving human-machine interaction time scales. 

\subsection{Discrete time representation of Lindblad Equation in terms of Kraus Operators}
\label{dt_deriv}
Our aim here is to derive \eqref{M_update} which is the discrete time version of \eqref{Lindblad}. This yields  the Lindbladian operator in terms of Kraus Operators used in Sec.\, \ref{sec3}. We will first define the Lindbladian operator and how it evolves. Then in Theorem I we prove a Lindbladian Operator can be reformulated in terms of Kraus Operators as done in \cite{Pearle_2012}.

Given a state of nature $s$ which is a random variable that can take on $n$ values, the human decision maker chooses one of $m$ possible actions. 
The quantum based model for human decision making is governed by the Lindbladian evolution of the psychological state. With $\mathcal{L}^u$ denoting the Lindbladian operator, the Lindbladian ordinary differential  equation for the dynamics of the psychological state $\rho_t$ over time $t \in [0,\infty)$  is specified as  
\begin{align}
\label{Lindblad}
&\frac{d\rho_t}{dt} = \mathcal{L}_{(\param,\posterior)}^u \, \rho_t, \quad \rho_0  = \frac{1}{nm}\diag(1,\cdots,1)_{mn\times mn} \\
     &\mathcal{L}_{(\param,\posterior)}^u \, \rho =-i(1-\parama)[H^{\control},\rho] \nonumber\\&+\parama \sum_{k,j}\gamma_{(k,j),(\paraml,\paramp)}^{\control}\left(L_{(k,j)}\,\rho\, L_{(k,j)}^{\dagger}-\frac{1}{2}\{L_{(k,j)}^{\dagger}L_{(k,j)},\rho\}\right)  \nonumber
\end{align}
Here $[A,B] = AB - BA$ is the commutator of $A$ and $B$, $\{A,B\} = AB + BA$ is the anti-commutator of $A$ and $B$, $A^*$ is the complex conjugate of $A$, and $k,j \in \{1,\dots,nm\}$. Also, $\param = (\parama,\paraml,\paramp) \in [0,1]\times[0,1]\times[0,\infty)$ is a vector of free parameters which determine the quantum decision maker's behavior, as discussed in Appendix \ref{constr_lind}. We  assume  that the machine has full knowledge of these behavioral parameters, as methods for estimating these via training are outside the scope of this paper. The terms $\param$, $L$, $\gamma$ and $H$ will be defined in Appendix \ref{constr_lind} in the context of an appropriate psychological model. Eq.~\eqref{Lindblad} is the canonical form of the Lindbladian operator for open-quantum systems.

The first term $[H^{\control},\rho]$ is known as the Von-Neumann equation and generalizes Schr\"{o}dinger's Equation from the pure state domain to the density operator domain.
The second term captures non-reversible Markovian dissipative effects of an interaction with an external environment, making this system "open". 
These two terms are weighted by the coefficient $\parama$, which interpolates between the purely quantum von Neumann evolution ($\parama = 0$) and the purely dissipative dynamics ($\parama = 1$).

 In order to prove convergence of the psychological state $\rho$ under an appropriate control policy, done in Sec.\, ~\ref{sec3}, we need to derive the discrete time representation  (\ref{M_update}) from the differential equation~\eqref{Lindblad}.
 
\begin{theorem}
\label{LindbladKraussProof}
The psychological state update (\ref{M_update}) is obtained as the first-order Euler time-discretization of the Lindbladian differential equation~\eqref{Lindblad}.
\end{theorem}

\begin{proof}
 See Appendix \ref{proof_lind}
\end{proof}

\subsection{Practicality in Modeling Human Decision Making}
\label{model_practicality}
The Lindbladian model \eqref{Lindblad}
captures important human decision features such as violation of the sure-thing principle, temporal preference oscillations and order effects, which we now describe.
These features cannot be explained by purely Markovian models in a parsimonious way \cite{busemeyer2009empirical}.

\subsubsection*{The violation of the sure-thing principle}
The total probability law, also called the Sure Thing Principle (STP), is  
\[P(A)=P(B)\,P(A|B)+P(\bar{B})\,P(A|\bar B)\]
for events $A$ and $B$. Violation occurs when $=$ is replaced by either $<$ or 
$>$.
Pothos and Busemeyer \cite{pothos2009quantum} (see also \cite{khrennikov2009quantum}) review  empirical evidence for the violation of STP and show how quantum models can account for it by introducing quantum interference in the probability evolution. Here we demonstrate that the Lindbladian model in Appendix \ref{Lindblad_param} can model these violations.

    Consider the simple case comprising  two actions $\act_1, \act_2$ and two states $\CE_1,\CE_2$. Suppose $\Gamma$ is a non-degenerate posterior belief (strictly in the interior of the unit simplex) of the underlying state. The violation of the sure thing principle occurs when $P(\act_2|\Gamma)$ is not a convex combination of $P(\act_2|\CE_1)$ and $P(\act_2|\CE_2)$, i.e. 
    \[P(\act_2|\Gamma) \neq \epsilon\,P(\act_2|\CE_1) + (1-\epsilon)\,P(\act_2|\,\CE_2) \ \forall\, \epsilon \in (0,1)\]
     Equivalently, if $P(\act_2|\Gamma) > \max\{P(\act_2|\CE_1),P(\act_2|\CE_2)\}$ or $P(\act_2|\Gamma) < \min\{P(\act_2|\CE_1),P(\act_2|\CE_2)\}$ then the sure thing principle has been violated. This latter form is explicitly demonstrated in Fig.~\ref{fig:STP_viol_surf}: The probabilities of $\act_2$ $P(\act_2|\CE_1)$ (green surface) $P(\act_2|\CE_2)$ (blue surface), and $P(\act_2|\Gamma)$ (red surface) are plotted as functions of free parameters $\parama$ and $\paramp$, with $\paraml = 10$ fixed. The regions of the plot where the red surface is outside the space in between the blue and green surfaces correspond to parameterizations under which the sure thing principle has been violated. Fig.~\ref{fig:STP_viol_plot} shows a cross-section of Fig.~\ref{fig:STP_viol_surf} at $\parama = 0.29$ and explicitly shows the violation of total probability for $\paramp > 0.88$.
    \begin{figure}[h!]
    \centering
    \includegraphics[width=1\linewidth]{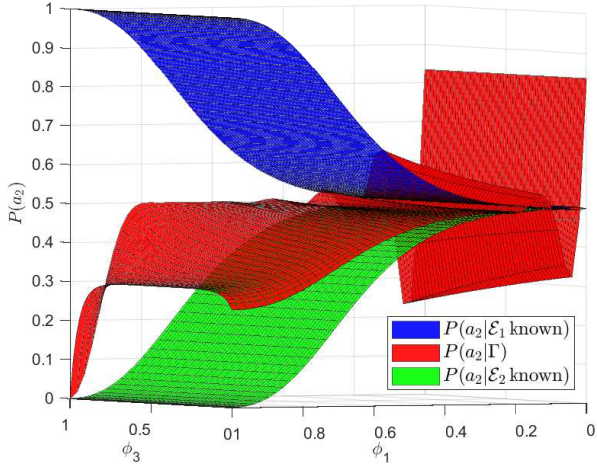}
    \caption{Quantum decision making facilitates modeling human subjective probabilities that violate the total probability law (sure thing principle). The surface for $P(\act_2|\,\Gamma)$, denoted in red crosses the surfaces of $P(a_2|\CE_1 \,\text{known})$ denoted in blue and $P(a_2|\CE_2\,\text{known})$ denoted in green. The law is not violated if the red surface lies between the green and the blue surface.}
    \label{fig:STP_viol_surf}
\end{figure}
\begin{figure}[h!]
    \centering
    \includegraphics[width=1\linewidth]{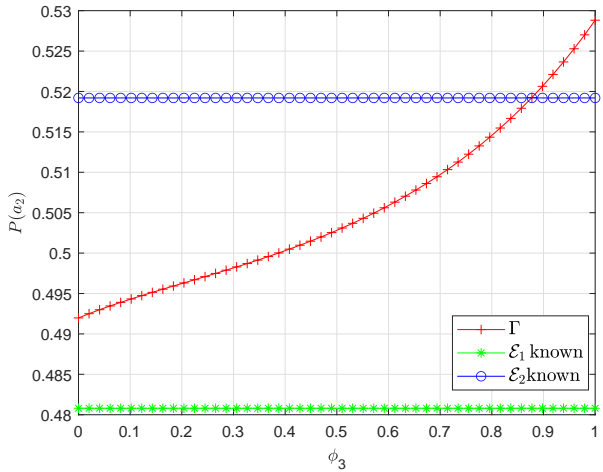}
    \caption{Violation of the law of total probability; cross-sectional view. This is a cross-section of the Fig.\ref{fig:STP_viol_surf} taken at $\parama=0.2857$. The violation occurs for $\phi_3 \in [0.87,1].$}
    \label{fig:STP_viol_plot}
\end{figure}

\subsubsection*{Oscillation in preference evolution}
Fig.~\ref{fig:time_evol} illustrates a sample path of the action preference distribution evolution governed by \eqref{M_update}. This sample path in particular demonstrates two essential features of the Lindbladian model:
\begin{itemize}
    \item Preferences which oscillate over time can be modeled. These oscillations in temporal preference evolution are empirically observed in \cite{kvam2021temporal} and cannot be modeled by classical Markovian models.
    \item Since the Lindbladian model includes dissipative (Markovian) dynamics, it does not oscillate forever; it converges to a steady-state distribution. This steady-state distribution is commonly used in the literature \cite{busemeyer2020application} \cite{martinez2016quantum}; Here we use it to generate the probability data in Figs. 2 and 3 which reveals violations of the sure thing principle.
\end{itemize}
\begin{figure}[h!]
    \centering
    \includegraphics[width=1\linewidth]{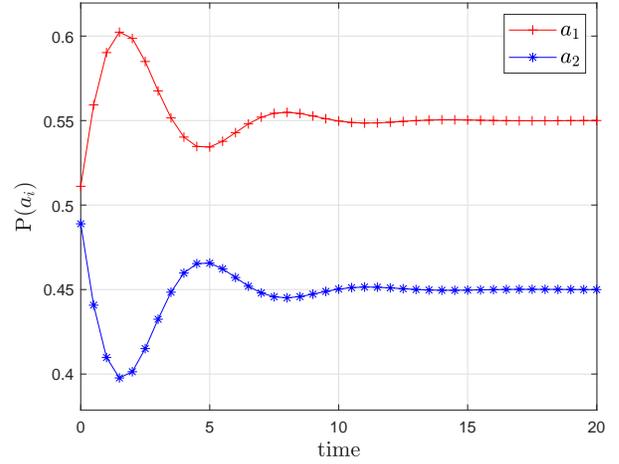}
    \caption{Time evolution of action probabilities exhibits oscillatory behavior. This behaviour cannot be captured  by the classical Markovian model, which only exhibits monotonic behavior. Here the parameters are $\parama=0.5$, $\paraml=10$ and $\paramp=0.5$. }
    \label{fig:time_evol}
\end{figure}

\section{Machine Control of Human Decision Maker}
\label{sec3}

The purpose of this section is to design a stochastic Lyapunov-based controller. By using this controller, our main result (Theorem~\ref{thm:convg}) shows that regardless of the initial psychological state of the human, the machine is  able to control the human preference in such a way that the target action is eventually chosen at every time step with probability one. Theorem~\ref{thm:convg} generalizes the \ly formulation in \cite{AMINI20132683} to demonstrate almost sure convergence when the machine interacts with the human over random time intervals $T$. 

We first define some notation:
With $ d =nm$, let $D$ denote the space of non-negative Hermitian matrices with trace~1:
\begin{equation}
\label{Dset}
D := \{\rho \in \complex^{d \times d} : \rho = \rho^{\dagger}, \Tr(\rho) = 1, \rho \geq 0\}
\end{equation}
Let $\{\ket{b_r}\}_{r=1}^d$ be a set of orthonormal vectors in $\complex^d$ (isomorphic to $\CH(\CA) \otimes \CH(\CX)$), where each $\ket{b_r}$ corresponds to a unique state-action pair. Let $\CS$ be the Hilbert space spanned by the orthonormal basis $\{\ket{b_r}\}_{r=1}^d$. Thus for each $a \in \CA$ there exists a unique subset $\{\ket{b_{i_1}},\dots,\ket{b_{i_n}}\} = a \otimes \CX$, and for each $s \in \CX$ there exists a unique subset $\{\ket{b_{r_1}},\dots,\ket{b_{r_m}}\} = \CA \otimes s$
\begin{definition}
\label{targ_def}
$\ket{\bar{b}_r} \in \{\ket{b_r}\}_{r=1}^d$ is the target state which is chosen by the machine and to which the human is guided via the control input
\end{definition}
The difference between target \textit{action} and target \textit{state} should be clarified. The target action $\act^{*} \in \CA$ lies in the action space and corresponds to one particular decision the human can make. The target \textit{state} $\ket{\bar{b}_r} \in \act^{*} \otimes \CX$ is one particular state in the Hilbert space $\CS$ consisting of both the action space $\CA$ and the state space $\CX$. We will show that the psychological state can be guided to any target \textit{state}, in particular some $\ket{\bar{b}_r} \in \act^{*} \otimes \CX$, which is equivalent to convergence to the target action $\act^{*}$.

The term {\em 'Open-loop (super) martingale'} below denotes a (super) martingale when  the control input $\control = 0$ for $k = 0,1,\ldots$. 

\subsection{Machine control}
\label{machine_control}
Recall from Sec. \ref{model} that the machine interacts with the human via a control signal $\control$.
We now introduce additional assumptions on this machine controller:
\begin{enumerate}
    \item
    {Control Input Constraints}\\
    We need the following  assumptions which constrain the influence of control $u_k$ on the psychological evolution.
\label{ap:contr}
\begin{enumerate}
    \item For each $\control, \sum_{\act_k}M_{\act_k}^{\control \dagger}\, M_{\act_k}^{\control} = \id$ for actions $\act_k$.
    \item For $\control = 0$, the $d\times\,d$ matrices $M_{\act_k}^0$ are diagonal: $ M_{\act_k}^0 = \sum_{n=1}^d c_{\act_k,n}\ket{n}\bra{n}, c_{\act_k,n} \in \complex$
    \item All $M_{\act_k}^{\control}$ are $\textrm{C}^2$ (continuously differentiable) functions of $\control$
\end{enumerate}
    \item 
    \label{3Alast}
    The control signal $\control$ is chosen according to\footnote{The minimum exists since $V_{\epsilon}(\cdot)$ is continuous and differentiable, from Sec. \ref{machine_control} 1.~(c), see~\cite{AMINI20132683}.}
    \begin{equation}
    \label{eq:Control_Input}
        \control := \underset{u \in [-1,1]}{\text{argmin}}\, \EX[V_{\epsilon}(\rho_{k+T}) | \rho_k, u, \DLsignal,\zeta]
    \end{equation}
    Here $V_{\epsilon}(\rho)$ will be  defined in (\ref{Ly_fcn}) and is parameterized by the set $\{\sigma_r\}_{r \in \{1,\dots,d\}}$ (see Appendix \ref{ap:sigma}). 
\end{enumerate}

{\em Discussion.}
The constraints (a)-(c) are needed as technical assumptions for the quantum control methods of \cite{AMINI20132683}, and thus for our proof of Theorem~\ref{thm:convg} to hold. For our purposes it suffices that $\control \in [-1,1]$, see \cite{AMINI20132683} for details. For our particular purposes, we are deriving the structure of $M_{\mu}^{\control}$ by the psychologically motivated Lindbladian structure given in \cite{martinez2016quantum}. So in order to utilize the stochastic Lyapunov stability methods of \cite{AMINI20132683}, we must adapt the Lindbladian structure (parameterization by $\control$) such that constraints (a)-(c) are satisfied.  Appendix \ref{Lindblad_param} discusses how this is accomplished. 

Constraint (a) is satisfied since the Lindbladian operator is a Kraus operator (see Sec.~\ref{dt_deriv}, Theorem~\ref{LindbladKraussProof}). Constraint (b) implies that Lindbladian operator is a diagonal operator at $\control=0$; it is shown by our formulation of the Lindbladian operator which is indeed a diagonal operator when $\control=0$, as given in Appendix \ref{Lindblad_param}. The constraint (c) relates to   the continuity of the Lindbladian operator with respect to $\control$ and is satisfied by our construction given in Appendix \ref{Lindblad_param}. This continuity property will be used to derive the parameters of the \ly function in Appendix \ref{ap:sigma}.  

We note that the choice of target state (Definition \ref{targ_def}) is arbitrary in the context of our main result, i.e. any particular target state can be chosen by the machine. For instance, we could have the machine determine the target state $\ket{\bar{b}_r} \in a^{*} \otimes \CX$, with corresponding action $a^{*}$ as that which optimizes the Bayesian expected objective utility:
\begin{equation}
    \label{contr_EUT}
        a^{*} = \textrm{argmax}_{a \in \{1,\dots,m\}} \EX[\zeta(a,\CE) | \DLsignal]
    \end{equation}
    where $\zeta$ is the objective state-action utility function. As the human has a utility which is augmented by subjective biases (see Appendix \ref{Lindblad_param}) and a psychological evolution which is inherently probabilistic, (\ref{contr_EUT}) would essentially improve the human's decision making such that it is Bayesian optimal. However, note that this particular choice of \textit{which} target action to guide the human to does not factor into the proof of our result that the human can be guided to \textit{any} chosen target action eventually.

 \subsection{Main Result 1. Convergence of Psychological State}
 \label{Sec:Qly}
 Here we prove the convergence of the controlled psychological state to a desired target action. 

\begin{theorem}
\label{thm:convg}
    Consider the discrete time density operator evolution (\ref{M_update}) and any target state $\ket{\bar{b}_r}$ (see Definition~\ref{targ_def}). Then, there exists a control sequence $\{\control\}_{k \in \mathbb{N}}$ generated by the machine such that the human psychological state $\rho_k$ converges to $\ket{\bar{b}_r}\bra{\bar{b}_r}$ with probability one for any initial $\rho_0 \in D$, where $D$ is defined in \eqref{Dset}. \\
\end{theorem}
\textit{Remark}: The convergence to any psychological state $\ket{\bar{b}_r}\bra{\bar{b}_r}, \ket{\bar{b}_r} \in a \otimes \CX$, implies the convergence to the corresponding action $a \in \CA$. Since $\ \forall a \in \CA$ there exists a set $\{\ket{b_{1}},\dots,\ket{b_{n}}\} = a \otimes \CX$ (see discussion directly above Definition \ref{targ_def}), Theorem~\ref{thm:convg} implies that the machine can guide the human to choose any specific target action $a \in \CA$ asymptotically.\\
\begin{proof}
First with $\joint = \{a_k, T\}$, rewrite (\ref{M_update}) as 
\begin{equation}
\label{eq:dens_up}
    \rho_{k+T} = \mathbb{M}_{\joint}^{\control}\rho_k = \frac{M_{\joint}^{\control}\, \rho_k\, M_{\joint}^{\control \dagger}}{\Tr(M_{\joint}^{\control}\, \rho_k\, M_{\joint}^{\control \dagger})}
\end{equation}
 We define the following \ly function, which forms a supermartingale under both open-loop (zero-input) and closed-loop (feedback control ($\control$)) conditions for the process (\ref{eq:dens_up}):
\begin{equation}
\label{Ly_fcn}
    V_{\epsilon}(\rho):= \sum_{r=1}^d \sigma_r \bra{b_r} \rho \ket{b_r} - \frac{\epsilon}{2}\sum_{r=1}^d \bra{b_r} \rho \ket{b_r}^2
\end{equation}
where $\sigma_r$ is non-negative $\forall r \in \{1,\ldots,d\}$ and $\epsilon$ is strictly positive. The set $\{\sigma_r\}_{r=1}^d$ and $\epsilon$ are chosen according to \cite{AMINI20132683} such that the \ly function $V_{\epsilon}(\rho_k)>0 \, \forall\, \rho_k \in D$ and $\rho_k$ converges to the intended subspace $\ket{\bar{b}_r} \bra{\bar{b}_r}$ with probability 1. By \cite{AMINI20132683} and \cite{6160433}, $\bra{b_r} \rho \ket{b_r}$ is an open-loop martingale given the $T$-step density operator evolution (\ref{eq:dens_up}) (see Appendix~\ref{ap:martingale} for proof). $V_{\epsilon}$ is a concave function of the open-loop martingales $\bra{b_r} \rho \ket{b_r}$ and therefore is an open-loop ($\control = 0$) supermartingale given the process (\ref{eq:dens_up}). 
\begin{equation*}
    \EX[V_{\epsilon}(\rho_{k+T}) |\, \rho_k, \control = 0] - V_{\epsilon}(\rho_k) \leq 0 
\end{equation*}
The following feedback control mechanism is used
\begin{equation}
\label{pf_contr}
    \control := \underset{u \in [-1,1]}{\text{argmin}}\, \EX[V_{\epsilon}(\rho_{k+T}) | \rho_k, u]
\end{equation}
to get $\EX[V_{\epsilon}(\rho_{k+T}) | \rho_k, \control] \leq \EX[V_{\epsilon}(\rho_{k+T}) | \rho_k, u = 0]$.
Here the expectation is taken with respect to $\joint$,
\begin{equation*}
\begin{split}
     \EX[V_{\epsilon}(\rho_{k+T}) | \rho_k, u]  &=
    \EX[V_{\epsilon}(\mathbb{M}_{\joint}^{u} \rho_k)]
=  \int_{\Omega} V_{\epsilon} (\mathbb{M}_{\joint(\omega)}^u \rho_k)d\omega 
\end{split}
\end{equation*}
where $\Omega$ is the sample space under which the process is induced.
Define 
\begin{equation*}
\begin{split}
    \tilde{Q}(\rho_k) := \EX[V_{\epsilon}(\rho_{k+T}) | \rho_k, \control] - V_{\epsilon}(\rho_k) \\ \leq \EX[V_{\epsilon}(\rho_{k+T}) | \rho_k, \control = 0] - V_{\epsilon}(\rho_k) \leq 0
\end{split}
\end{equation*}
    
$V_{\epsilon}(\rho)$ is a continuous function, so using \cite[Chapter 8]{kushner1971introduction}, so we have the $T$-step control sequence $\{\rho_{k+iT}\}_{i \in \mathbb{N}}$ that converges to the set $\textrm{\limset}:= \{\rho : \tilde{Q}(\rho) = 0\}$ with probability one.

By Lemma~\ref{limset_restr}, presented immediately after this proof, the set \limset \ is restricted to our desired state $\{\ket{\bar{b}_r} \bra{\bar{b}_r}\}$. So now we must show that the entire sequence $\{\rho_k\}_{k \in \mathbb{N}}$ converges to this set. 

We now show that $\PR[\lim_{k \to \infty}\rho_k = \ket{\bar{b_r}}\bra{\bar{b_r}}] = 1$, i.e. the entire sequence converges to the target state with probability one. This utilizes Theorem~\ref{thm:extension} which is developed in Sec. \ref{sec4}.

The statement of Lemma~\ref{limset_restr} can be rewritten as 
\begin{equation*}
\begin{split}
&\PR[\,\lim_{i \to \infty} \rho_{k + iT} = \ket{\bar{b}_r}\bra{\bar{b}_r}] = \\
&\PR[\,\exists\,K \in \mathbb{N} : \rho_{k + iT} = \ket{\bar{b}_r}\bra{\bar{b}_r} \ \forall i \geq K]=1    
\end{split}
\end{equation*}
where the choice of $\{\sigma_r\}_{r=1}^d$, from Lemma 2 of \cite{AMINI20132683} and our proof of Lemma~\ref{limset_restr}, ensures that 
\[\EX[V_{\epsilon}(\rho_{k+T}) \ | \rho_k = \ket{b_r}\bra{b_r}, \control = u]\] has a strict local minimum at $u=0$ for $r = \bar{r}$, so
\begin{equation*}
\begin{split}
    u_{k+iT}= \underset{u \in [-1,1]}{\text{argmin}}\, \{\EX[V_{\epsilon}(\rho_{k+(i+1)T}) | \rho_{k+iT}, u]\} = 0 \, \forall\, i \geq K
\end{split}
\end{equation*}
Now, for $q \in \{1,\dots,T\}$, 
\begin{equation*}
\begin{split}
    &\EX[V_{\epsilon}(\rho_{k+(i+q+1)T}) | \rho_{k+(i+q)T}, u=0] \\&=  \EX[V_{\epsilon}(\mathbb{M}_{\joint}^{0}\rho_{k+(i+q)T})]\leq V_{\epsilon}(\rho_{k+(i+q)T})
\end{split}
\end{equation*}
since $V_{\epsilon}$ is an open-loop supermartingale. This implies $\exists\, K \in \mathbb{N}$ such that 
\[\EX[V_{\epsilon}(\rho_{k+T}) | \rho_k, \control] \leq V_{\epsilon}(\rho_k) \ \forall\, k \geq K\] 
Since there exists a unique mapping from elements of $D$ to elements of $\reals^{2d^2}$, applying Theorem~\ref{thm:extension} yields $\PR[\lim_{k \to \infty} \rho_k = \ket{\bar{b}_r}\bra{\bar{b}_r}] = 1$, thus proving the Theorem.

\end{proof}

\begin{lemma}
\label{limset_restr}
For the evolution \eqref{eq:dens_up}, target state $\ket{\bar{b}_r}$ and feedback control \eqref{pf_contr}, we have
$\PR[\,\lim_{i \to \infty} \rho_{k + iT} = \ket{\bar{b}_r}\bra{\bar{b}_r}] = 1$
\end{lemma}

\begin{proof}
The proof follows immediately from Lemma 2 of~\cite{AMINI20132683}. For any target subspace $\ket{\bar{b}_r} \bra{\bar{b}_r}$, the set $\{\sigma_r\}_{r=1}^d$ can be chosen in such a way that \limset $ = \ket{\bar{b}_r} \bra{\bar{b}_r}$. The idea is the following:
A state $\rho_k$ is in the limit set \limset \ iff for all $u \in [-1,1],$
\begin{equation}
\label{supm}
    \EX[V_{\epsilon}(\rho_{k+T}) | \rho_k, u] - V_{\epsilon}(\rho_k) \geq 0
\end{equation}
Also, since $V_{\epsilon}$ is an open-loop supermartingale, $\forall \rho_k \in D$: 
\begin{equation}
\label{subm}
    \EX[V_{\epsilon}(\rho_{k+T}) | \rho_k, u = 0] - V_{\epsilon}(\rho_k) \leq 0
\end{equation}
By Lemma 2 of \cite{AMINI20132683}, presented in Appendix \ref{ap:sigma} as Lemma~\ref{ap:lem:sig2}, given any $\bar{r} \in \{1,\dots, d\}$ and fixed $\epsilon > 0$, the weights $\sigma_1, \dots, \sigma_d$ can be chosen so that $V_{\epsilon}$ satisfies the following property: $\forall r \in \{1,\dots,d\}, u \mapsto \EX[V_{\epsilon}(\rho_{k+T}) \ | \rho_k = \ket{b_r}\bra{b_r}, \control = u]$ has a strict local minimum at $u=0$ for $r = \bar{r}$ and a strict local maximum at $u=0$ for $r \neq \bar{r}$. This combined with (\ref{subm}) ensures that for any $r \neq \bar{r}, \exists\, u \in [-1,1]$ such that 
\[\EX[V_{\epsilon}(\rho_{k+T}) \ |\, \rho_k = \ket{b_r}\bra{b_r}, \control = u] - V_{\epsilon}(\ket{b_r}\bra{b_r}) < 0\]
Therefore, by (\ref{supm}), we have that $\ket{b_r}\bra{b_r}$ is in the limit set $l_{\infty}$ if and only if $r = \bar{r}$.
\end{proof}

To summarize,  we showed that for the discrete time psychological state evolution of (\ref{M_update}), there exists a control policy which allows the machine to guide the human decisions such that a target decision is made eventually, almost surely.

\subsection{Numerical Example. Controlling Human Decision Evolution}

    We now provide numerical examples which illustrate the time evolution  of the action distribution \eqref{eq:ProbEvo} and target action convergence. In each case we initialize the psychological state $\rho_0$ as a uniform (non-informative) distribution over the action space. The psychological state evolution is given by \eqref{M_update}, and the objective state-action utility $\zeta(\act_i,\CE_j)$ in \eqref{zeta} is initialized to have uniformly distributed random values between 0 and~10.
    
    
\subsubsection{Example 1}
Suppose there are two underlying states of nature $\CE_1$ and $\CE_2$, and four actions $\act_1$, $\act_2$, $\act_3$, $\act_4$. The target action is $\act_4$. We use parameters $\parama = 0.8$ and $\paraml = 10$. We simulated the preference evolution \ref{M_update} according to the protocol (\ref{2Afirst})-(\ref{2Alast}) of Sec. \ref{model} and (\ref{ap:contr})-(\ref{3Alast}) of Sec. \ref{machine_control}. The action preference distribution \eqref{eq:ProbEvo} is plotted as a function of time, with actions and controls taken every $T=10$ time points. A typical sample path is shown in Fig.~\ref{fig:targ_action4}. At each 10 second interval, the chosen action is indicated by the curve that jumps to have probability one. Observe that after time 60 (6 machine interventions) the preference distribution converges to the target action $\act_4$. Thus, the example illustrates the convergence asserted in Theorem~\ref{thm:convg}.

\subsubsection{Example 2}
The human-machine interaction can also yield oscillatory behavior in the preference distribution sample path before convergence; yet, Theorem~\ref{thm:convg} guarantees convergence.\footnote{Previously we demonstrated this oscillatory behavior for the uncontrolled case (Fig.~\ref{fig:time_evol}).}
We have two underlying states of nature $\CE_1$ and $\CE_2$, and two actions $\act_1$, $\act_2$. The target action is $\act_1$. We use parameters $\parama = 0.25$ and $\paraml = 10$. The deliberate choice of a lower value of $\parama$ compared to $\parama = 0.8$ in example 1 induces oscillatory behavior.  Observe that after time 30 (3 machine interventions) the preference distribution converges to the target action $\act_1$.
\begin{figure}[h!]
    \centering
    \includegraphics[width=1\linewidth]{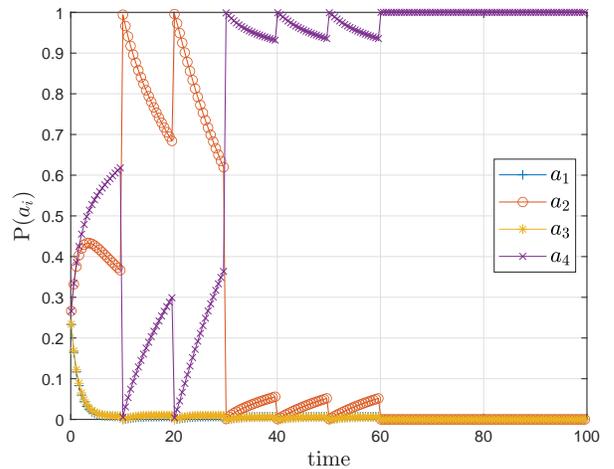}
    \caption{Simulation illustrating how the machine controls the human psychological state to converge to action 4. In this simulation the human decision maker's action probabilities converge after $T=60$ time steps.}
    \label{fig:targ_action4}
\end{figure}
\begin{figure}[h!]
    \centering
    \includegraphics[width=1\linewidth]{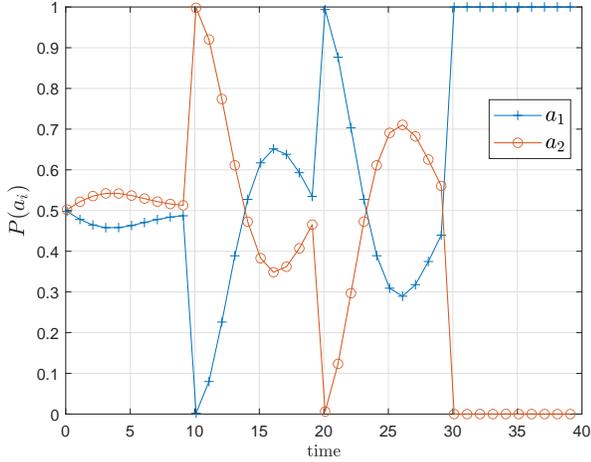}
    \caption{Simulation illustrating how the machine controls the human psychological state to converge to action 1. In this simulation the human decision maker's action probabilities converge after $T=30$ time steps. This example showcases oscillatory behavior in the preference distribution evolution. Such oscillatory behavior cannot be captured by classical Markovian models, as discussed in Sec. \ref{model_practicality}.}
    \label{fig:targ_action5}
\end{figure}

\section{Finite Step Stochastic \ly Stability}
\label{sec4}
Recall in Sec.~\ref{sec3}, our stability proof of the human-machine decision system (Theorem~\ref{thm:convg}) employed a finite-step stochastic \ly argument, the proof of which was postponed to this section. The purpose of this section is two-fold: to generalize an existing finite-step stochastic \ly result given in \cite{qin2019lyapunov}, and to use a special case of this generalization as this argument which enables the proof of our main result, Theorem~\ref{thm:convg}.

This section presents three theorems: 
\begin{itemize}
\item Theorem~\ref{thm:qin} is equivalent to Theorem 1 of \cite{qin2019lyapunov}, but we provide an alternative proof which allows for direct generalization.
\item Theorem~\ref{thm:rand} generalizes Theorem~\ref{thm:qin} to the case when the finite-step size $T$ is a random variable
\item Theorem~\ref{thm:extension} follows as a special case of Theorem~\ref{thm:rand} which is directly applicable to our proof of Theorem~\ref{thm:convg}.
\end{itemize}

Consider the discrete time stochastic system described by 
\begin{equation}
\label{eq1}
    x_{k+1} = f(x_k, y_{k+1}), \  k= 0,1,2,\ldots
\end{equation}
Here $x_k \in \reals^n$, and $\{y_k : k =0,1,2,\ldots\}$ is a $\reals^d$ valued stochastic process on the probability space $(\Omega, \mathcal{F}, \PR)$.
Consider the filtration (increasing sequence of $\sigma$-fields)  defined by $\mathcal{F}_0 = \{\emptyset, \Omega\}, \ \mathcal{F}_k = \sigma(y_1, \dots, y_k) \textrm{ for } k \geq 1$. We choose $x_0 \in \reals^n$ as a  constant with probability one. Thus $\{x_n\}$ is an $\reals^n$-valued stochastic process adapted to $\mathcal{F}_k$. \\

\begin{theorem}
\label{thm:extension}
For the discrete-time stochastic system (\ref{eq1}), let $V: \reals^n \rightarrow \reals$ be a continuous non-negative and radially unbounded function. Suppose there exists a random variable $T: \Omega \rightarrow~\mathbb{N}$ such that for any $k$, 
\begin{equation}
    \EX[V(x_{k+T}) | \mathcal{F}_k] - V(x_k) \leq -\varphi(x_k)
\end{equation}
where $\varphi:~ \reals^n\rightarrow~\reals$ is continuous and satisfies \[\varphi(x) \geq 0 \ \  \forall x \in \reals^n\]
Then for any initial condition $x_0 \in \reals^n$, $x_k$ converges to \\$\mathcal{D}_1 := \{x \in \reals^n : \varphi(x) = 0\}$ with probability one. 
\end{theorem}

\begin{proof}
The proof follows from Theorems \ref{thm:qin} and \ref{thm:rand} below. 
\end{proof}

\begin{theorem}
\label{thm:qin}
For the discrete-time stochastic system (\ref{eq1}), let $V: \reals^n \rightarrow \reals$ be a continuous non-negative and radially unbounded function. Define the set $Q_{\lambda} := \{x : V(x) < \lambda\}$ for some positive $\lambda$, and assume that:\\
\ \ \ (a) $\EX[V(x_{k+1}) | \mathcal{F}_k] - V(x_k) \leq 0$ for any $k$ such that $x_k \in Q_{\lambda}$ \\
\ \ \ (b) There exists an integer $T \geq 1$, independent of $\omega \in \Omega$, such that for any $k, \EX[V(x_{k+T}) | \mathcal{F}_k] - V(x_k) \leq -\varphi(x_k)$, where $\varphi: \reals^n \rightarrow \reals$ is continuous and satisfies $\varphi(x) \geq 0$ for any $x \in Q_{\lambda}$ \\ 
Then for any initial condition $x_0 \in Q_{\lambda}$, $x_k$ converges to $\mathcal{D}_1 := \{x \in Q_{\lambda} : \varphi(x) = 0\}$ with probability at least $1 - V(x_0)/\lambda$ \\
\end{theorem} 
\begin{proof}
From assumption (b), 
\begin{equation}
\label{eq2}
    \EX[V(x_{k+T})|\mathcal{F}_k] - V(x_k) \leq -\varphi(x_k) \leq 0 ,  \forall x_k \in Q_{\lambda}
\end{equation} where $\varphi(x)$ is continuous. Now, by assumption (a) and \cite[p.196]{kushner1971introduction}, if we start with $x_0 \in Q_{\lambda}$ then the probability of staying in $Q_{\lambda}$ during the entire resultant trajectory is at least $1 - V(x_0)/\lambda$, i.e. 
\begin{equation}
\label{eq3}
    \PR[\sup_{k \in \mathbb{N}}V(x_k) \geq \lambda] \leq V(x_0)/\lambda
\end{equation}
Next construct  $T$ subsequences of $\{X_k\}$ as follows: 
\[\{X_i^{(0)}\} = \{X_0,X_T,\dots\}, \{X_i^{(1)}\} = \{X_1,X_{T+1},\dots\},\dots\]
\[\{X_i^{(T-1)}\} = \{X_{T-1},X_{2T-1},\dots\}\] 
Suppose $\varphi(x) \geq 0 \ \forall x:$ Then for all $k \in K := \{0,\dots,T-1\}$, $\{V(X_i^{(k)})\}$ is a non-negative supermartingale process by (\ref{eq2}), and thus by Doob's convergence theorem converges to a limit with probability 1. From (\ref{eq2}) we have for all $k \in K$ and $n \in \mathbb{N}$
\[\sum_{l=1}^{n}\EX(V(X_l^{(k)}))-\EX(V(X_{l-1}^{(k)}))\leq-\EX(\sum_{l=0}^{n-1}\varphi(X_l^{(k)}))\]

and thus
\[0 \leq \EX(V(X_n^{(k)})) \leq \EX(V(X_0^{(k)})) - \EX(\sum_{l=0}^{n-1}\varphi(X_l^{(k)}))\]
We use Fatou's Lemma to obtain 
\[\EX(\sum_{l=0}^{\infty}\varphi(X_l^{(k)})) < \infty\] 
and by the  Borel-Cantelli lemma we have 
\[\PR[\lim_{n \to \infty}\varphi(X_n^{(k)}) = 0] = 1\, \forall\, k \in K
\]
Now suppose $\varphi(x) \geq 0$ only for $x \in Q_{\lambda}.$ Stop $\{X_n^{(k)}\}$ on first leaving $Q_{\lambda}$. Then for $x \notin Q_{\lambda}$, $\varphi(x) = 0$ for this stopped set. This stopped process is a supermartingale and the proof is the same as above.  \\
It is now apparent that 
\begin{equation}
    \lim_{n \to \infty}\varphi(X_n^{(k)}(\omega)) = 0 
    \label{limitConv}
\end{equation} 
$\, \forall\, k \in K\,\omega \in \bar{\Omega} = \{\omega \in \Omega : x_n(\omega) \in Q_{\lambda} \ \forall\, n \in \mathbb{N}\}$ so we have \[\PR[\lim_{n \to \infty}\varphi(X_n(\omega)) = 0] \geq 1- V(x_0) / \lambda\] by the  analysis in  Appendix \ref{pf:subseq_conv} and (\ref{eq3}). \\
\end{proof}

We now present a generalization of Theorem~\ref{thm:qin} for the case when the finite-step length $T$ is a random variable.  
\begin{theorem}
\label{thm:rand}
For the discrete-time stochastic system (\ref{eq1}), let $V: \reals^n \rightarrow \reals$ be a continuous non-negative and radially unbounded function. Define the set $Q_{\lambda} = \{x : V(x) < \lambda\}$ for some positive $\lambda$, and assume that:\\
\ \ \ (a) $\EX[V(x_{k+1}) | \mathcal{F}_k] - V(x_k) \leq 0$ for any $k$ such that $x_k \in Q_{\lambda}$ \\
\ \ \ (b) There exists a random variable $T: \Omega \rightarrow \mathbb{N}$, such that for any $k,$ \[ \EX[V(x_{k+T}) | \mathcal{F}_k] - V(x_k) \leq -\varphi(x_k)\] where $\varphi: \reals^n \rightarrow \reals$ is continuous and satisfies $\varphi(x) \geq 0$ for any $x \in Q_{\lambda}$ \\ 
Then for any initial condition $x_0 \in Q_{\lambda}$, $x_k$ converges to $\mathcal{D}_1 :=~\{x \in Q_{\lambda} : \varphi(x) = 0\}$ with probability at least $1 - V(x_0)/\lambda$ 
\end{theorem}
\textit{Remark}: Note that Theorem~\ref{thm:rand} is equivalent to Theorem~\ref{thm:qin} when the condition that $T$ is  independent of $\omega \in \Omega$ is removed. 

\begin{proof} 
Let $\{X_i\}_{n \in \mathbb{N}}$ be the original random process corresponding to the stochastic system \eqref{eq1} initialized with $x_0 \in Q_{\lambda}$. Construct the random process $\{Y_i\}_{i \in \mathbb{N}}$ by inducing the same stopped process on $\{X_i\}$ as in the proof of Theorem~\ref{thm:qin}. This simplifies analysis by omitting sample paths which leave the set $Q_{\lambda}$. Thus we consider only the sample paths which remain in $Q_{\lambda}$ for the entire trajectory and we will impose a bound on the measure of this set below.\footnote{This procedure is done implicitly in the proof of Theorem~\ref{thm:qin}.}
Now define a new random process $\{Y_{k_i}^{k_0}\}_{i \in \mathbb{N}}$ as follows: \\
Starting with $Y_{k_0}, k_0 \in \mathbb{N}$, define $\forall \omega \in \Omega: Y_{k_{1}}^{k_0} (\omega) = Y_{k_0 + T(\omega)}(\omega), Y_{k_{2}}^{k_0} (\omega) = Y_{k_0 + 2T(\omega)}(\omega) , \dots$  \[Y_{k_N}^{k_0} = Y_{k_0 + NT(\omega)}(\omega)\]i.e. take a subsequence (over multiples of $T(\omega)$) of the original sequence initialized at $Y_{k_0}, k_0 \in \mathbb{N}$. Then by (b) we have:
\begin{multline}
    \EX[V(Y_{k_{i+1}}^{k_0})|\mathcal{F}_{k_i}] - V(Y_{k_i}^{k_0}) \leq -\varphi(Y_{k_i}^{k_0}) \leq 0 , \\ \forall Y_{k_i}^{k_0} \in Q_{\lambda}, \ i \in \mathbb{N}.
\end{multline}
where $\mathcal{F}_{k_i} = \sigma(\{Y_n\}_{n=0}^{k_i})$.
Thus by the same method as in the proof of Theorem~\ref{thm:qin}, we can obtain
\begin{equation*}
    0 \leq \EX(V(Y_{k_{i+1}}^{k_0})) \leq \EX(V(Y_{k_0})) - \EX(\sum_{l=0}^{i}\varphi(Y_{k_l}^{k_0}))
\end{equation*}
and so by Borel-Cantelli:
\begin{equation}
\label{BorelCantelli}
    \PR[\lim_{i \to \infty}\varphi(Y_{k_i}^{k_0}) = 0] = 1 \ \ \forall k_0 \in \mathbb{N}
\end{equation}
where recall $\{Y_{k_i}^{k_0}\}_{i \in \mathbb{N}}$ is the sequence induced by the initial condition $Y_{k_0}$. Let $\Omega_{k_0} = \{\omega :  \lim_{i \to \infty}\varphi(Y_{k_i}^{k_0}) = 0\}$. We know from \eqref{BorelCantelli} that $\PR(\Omega_{k_0}) = 1 \ \forall k_0 \in \mathbb{N}$ and thus 
\[\PR(\{\omega :  \lim_{i \to \infty}\varphi(Y_{k_i}^{k_0}) = 0 \ \forall k_0 \in \mathbb{N}\})  = \PR(\bigcap_{k_0 \in \mathbb{N}}\Omega_{k_0}) = 1\]
Define \[\hat{\Omega} = \{\omega : \lim_{i \to \infty}\varphi(Y_{k_i}^{k_0}(\omega)) = 0 \ \forall k_0 \in \{1,\dots,T(\omega)\}\] 
Obviously $\bigcap_{k_0 \in \mathbb{N}} \Omega_{k_0}\,\subseteq\,\hat{\Omega}$ and so $\PR(\hat{\Omega}) = 1$, and $\forall \omega \in \hat{\Omega}$
the set of sub-sequences $\{Y_{k_i}^{k_0}(\omega)\}_{i \geq 1, k_0 = 1,\dots,T(\omega)}$ exhaust the entire sequence $\{Y_i(\omega)\}_{i \in \mathbb{N}}$ and thus $\lim_{i \to \infty}\varphi(Y_i(\omega)) = 0 \ \forall \omega \in \hat{\Omega}.$ Denote $\Omega^{*} = \{\omega : \lim_{i \to \infty}\varphi(Y_i(\omega)) = 0\} \supseteq \hat{\Omega}$. So 
$\PR(\Omega^{*}) = 1$, and observe that 
\[\Omega^{*} \cap \bar{\Omega} \subseteq \{\omega : \lim_{n \to \infty}\varphi(X_n(\omega)) = 0\}\]
so
\begin{equation*}
    \PR[\lim_{n \to \infty}\varphi(X_n(\omega)) = 0] \geq  \PR[\Omega^{*} \cap \bar{\Omega}] \geq 1- V(x_0) / \lambda
\end{equation*}
where $\bar{\Omega}$ is the same as defined in the proof of Theorem~\ref{thm:qin} and $\PR(\bar{\Omega}) \geq 1-V(x_0)/\lambda$ by assumption (a) and \cite[p.196]{kushner1971introduction}. \\
\end{proof}

Theorem~\ref{thm:qin} is an alternative proof of Theorem 1 in \cite{qin2019lyapunov} and Theorem~\ref{thm:rand} builds upon it to generalize $T$ to a random variable. Theorem~\ref{thm:extension} follows from Theorem~\ref{thm:rand} in the following way: Start with the assumptions of Theorem~\ref{thm:rand}. Now let $\phi(x) \geq 0 \ \forall x \in \reals^n$. Then $Q_{\lambda} = \reals^n$. Now since $V$ is radially unbounded, the only such $\lambda$ for which this can hold is $\lambda = \infty$. Then as a result of Theorem~\ref{thm:rand} we have: for any initial $x_0 \in Q_{\lambda} = \reals^n$, $x_k$ converges to $\mathcal{D}_1$ with probability at least $1-V(x_0)/\lambda = 1$, since $\lambda = \infty$.

To summarize, this section provided a generalization of the finite-step \ly function result of Qin. et. al. \cite{qin2019lyapunov}, by making $T$ a random variable. We applied this to show stability of the Lindbladian dynamics to prove almost sure convergence of the density operator (psychological state), but the generalization is of independent interest. 

\section{Conclusion and Extensions}
\label{sec5}
 Quantum decision models have emerged as a novel paradigm in psychology that can parsimoniously model special aspects of human decision making (such as order effects and sure thing principle) which  classical Markov models cannot. 
 In this paper, we formulated a human-machine interaction system as a {\em controlled} quantum decision making process. 
 We constructed a machine control procedure based on stochastic \ly stability which allows the psychological state of a human to be guided to any arbitrary target state eventually almost surely. This used a  random finite-step stochastic Lyapunov function result, which is  also of independent interest. 

The Lindbladian quantum  models that we used, capture  subjective biases and violations of rationality which cannot be captured by classical models; If a human's decision making is suboptimal due to these subjective irrational biases, then the machine can compensate and guide the human to be an optimal Bayesian expected utility maximizer (with respect to the machine's state measurement). 

We made  several simplifying assumptions that  warrant further investigation. We  assumed that the machine knows the parameterization of the human's Lindbladian  operator for the psychological state. It would be interesting to see what the analysis yields when these are only estimates with some distribution. We  considered the case when there is \textit{one} specified action for the human to be guided to. Generalization to classes of actions and associated  rate of convergence to the  target actions are interesting extensions.


\appendices
\section{Quantum Decision Model. Additional Details}
This Appendix discusses additional details of the  Quantum Decision model used in the paper. We draw heavily from   \cite{martinez2016quantum}.
\subsection{Open-Quantum System as a Psychological Model}
\label{model_derivation}
The derivation of the psychological evolution \eqref{Lindblad} uses a quantum random walk (QRW) over a state-action network, see \cite{martinez2016quantum} for details. The purpose of this subsection is to discuss this QRW formulation of \eqref{Lindblad}. Each node of the network represents a unique state-action pair. So, in our case there are a total $nm$ nodes. The network is decomposed into $m$ sub-graphs, representing underlying states. Each sub-graph has $n$ nodes and each node within this sub-graph has a self loop and is also connected to all the other nodes in the sub-graph bidirectionally, representing the action-action transition with in an underlying state. These edges are weighed according to objective state-action utility function $\utility$. These $m$ sub-graphs are connected via same actions across all underlying states. An example of this structure can be seen in Fig.~ \ref{fig:Network}, for the case of two underlying states and three actions $\act_1,\, \act_2,\, \act_3$. This network structure allows for the \textit{simultaneous} comparison of utilities and elicitation of beliefs in the underlying state.
\begin{figure}[h!]
    \centering
    \resizebox{0.48\textwidth}{!}{%
      \begin{tikzpicture}[scale=0.5, >=stealth']
    \begin{scope}[every node/.style={circle,thick,draw}]
        
        \node (A1_1) at (-9,-6) {$\act_1$};
        \node (A2_1) at (-3,-2) {$\act_2$};
        \node (A3_1) at (2,-6) {$\act_3$};
        
        \node (A1_2) at (-9,5) {$\act_1$};
        \node (A2_2) at (-3,1) {$\act_2$};
        \node (A3_2) at (2,5) {$\act_3$};
    \end{scope}
    
    \begin{scope}[>={Stealth[red]},
              every node/.style={fill=white,circle},
              every edge/.style={draw=red}]
             
             \path [->] (A1_1) edge[bend left=15]  (A3_1);
             \path [->] (A1_1) edge[bend left=15]  (A2_1);
             \path [->] (A1_1) edge[loop left=10]  (A1_1);
             \path [->] (A2_1) edge[bend left=15]   (A1_1);
             \path [->] (A2_1) edge[bend left=15]   (A3_1);
             \path [->] (A2_1) edge[loop left=1]  (A2_1);
             \path [->] (A3_1) edge[bend left=15]  (A1_1);
             \path [->] (A3_1) edge[bend left=15]  (A2_1);
             \path [->] (A3_1) edge[loop right=10]  (A3_1);
    \end{scope}
    
    \begin{scope}[>={Stealth[blue]},
              every node/.style={fill=white,circle},
              every edge/.style={draw=blue}]
             
             \path [->] (A1_2) edge[bend left=10]  (A3_2);
             \path [->] (A1_2) edge[bend left=10]  (A2_2);
             \path [->] (A1_2) edge[loop left=10]  (A1_2);
             \path [->] (A2_2) edge[bend left=10]   (A1_2);
             \path [->] (A2_2) edge[bend left=10]   (A3_2);
             \path [->] (A2_2) edge[loop left=10]  (A2_2);
             \path [->] (A3_2) edge[bend left=10]  (A1_2);
             \path [->] (A3_2) edge[bend left=10]  (A2_2);
             \path [->] (A3_2) edge[loop right=10]  (A3_2);
    \end{scope}

    \begin{scope}[>={Stealth[black]},
              every edge/.style={draw= black,densely dotted}]
            
            \path [->] (A3_2) edge (A3_1);
            \path [->] (A3_1) edge (A3_2);
            
            \path [->] (A2_2) edge (A2_1);
            \path [->] (A2_1) edge (A2_2);
            
            \path [->] (A1_2) edge (A1_1);
            \path [->] (A1_1) edge (A1_2);
        
    \end{scope}
     \node [ draw , rectangle , thick, dashed,red, minimum width = 9cm, 
    minimum height = 3.5cm] at (-3.5 ,-4.5) {$\CE_2$};
     \node [ draw , rectangle , thick, dashed,blue, minimum width = 9cm, 
    minimum height = 3.5cm] at (-3.5 ,3.5) {$\CE_1$};
  \end{tikzpicture} 
  }
    \caption{State-Action network over which the Quantum Random Walk takes place. $\CE_j$'s are the underlying states and $\act_i$'s are the actions.}
    \label{fig:Network}
\end{figure}
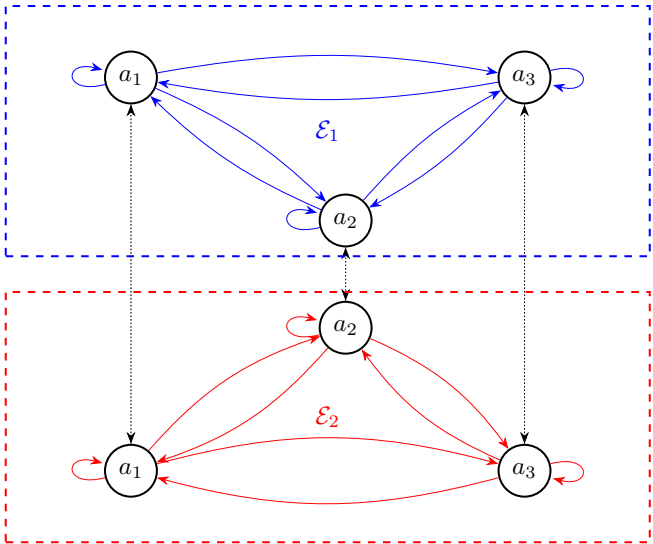
\subsection{Stochastic Quantum Random Walk. Lindbladian Dynamics}
\label{constr_lind}
We now relate the QRW in Appendix~\ref{model_derivation} to the Lindbladian Dynamics in \eqref{Lindblad}.
  This forms the basis of the decision making process \eqref{M_update} that  we control.
 \begin{enumerate}
     \item The state space  $\CX = \{\CE_1, \CE_2, \dots, \CE_N \}$. The Hilbert space over $\CX$ is defined as $\CH_{\CE} :=~\CH\{\ket{\CE_1},\dots,\ket{\CE_n}\}$. This space contains $\ket{\mathcal{E}_l}$, $n$-dimensional complex vectors indexed by the state $l$ 
     \item The action space  $\CA = \{\act_1,\act_2,\dots, \act_m\}$. The Hilbert space over $\CA$ is defined as $\CH_{a} :=~\CH\{\ket{a_1},\dots,\ket{a_m}\}$. This space contains $\ket{\act_j}$, $m$-dimensional complex vector indexed by the action $j$
     \item The psychological state space  $\CS := \CH_{\CE}\otimes \CH_{a}$. This space contains all action-state pairs. Each tuple $\ket{\CE_l,\,\act_j}$ represents each node in the graph given in Sec.~\ref{model_derivation}.
 \end{enumerate}
The psychological state $\rho$ is represented as a density operator in the Hilbert Space $\CS$, and its evolution is given by (\ref{Lindblad}). 
We also define non-negative functions $f(\control)$ and $g(\control)$ to be continuously differentiable and equal to zero for $\control = 0$.
\begin{enumerate}[resume]
\item The  Hamiltonian $H^u$ is defined as the block diagonal $nm~\times~nm$ symmetric matrix
\begin{align}
    \begin{split}
    \label{Hamiltonian}
        H^u := {\begin{bmatrix}
        \textbf{H} & \frac{\mathbf{I}\,f(u)}{z} & \dots & \frac{\mathbf{I}\,f(u)}{z}\\
        \frac{\mathbf{I}\,f(u)}{z} & \textbf{H} & \dots & \frac{\mathbf{I}\,f(u)}{z}\\
        \vdots & \vdots &  \ddots & \vdots\\
        \frac{\mathbf{I}\,f(u)}{z} & \frac{\mathbf{I}\,f(u)}{z} & \dots & \textbf{H}\\
      \end{bmatrix}}
    \end{split}
\end{align}
%
\[\textbf{H} := \begin{bmatrix}
        \frac{1}{z}(1-f(u)) & \frac{1}{z}f(u) & \dots & \frac{1}{z}f(u)) \\
        \frac{1}{z}f(u) & \frac{1}{z}(1-f(u)) & \dots & \frac{1}{z}f(u) \\
        \vdots & \vdots & & \vdots\\
        \frac{1}{z}f(u) & \frac{1}{z}f(u) & \dots & \frac{1}{z}(1-f(u)) \\
      \end{bmatrix} \]
with $z = (m+n-3)f(u) + 1$.
\end{enumerate}
$\parama$ determines a convex combination of the quantum and Markovian terms in \eqref{Lindblad}. A larger value of $\parama$ will result in more dissipative (less oscillatory) dynamics \cite{martinez2016quantum}. 
\begin{enumerate}[resume]
\item $L_{k,j}=\ket{k}\bra{j}$ represents the walk from $k^{th}$ action-state tuple to $j^{th}$ action-state tuple.
     \item We define $\gamma_{(k,j)}^u$ as the $(k,j)^{th}$ element in the \textit{Cognitive Matrix}
\begin{align}
\label{cog_mat}
\begin{split}
     \gamma_{(k,j)}^u&:=[C(\paraml,\paramp,u)]_{k,j} \\&= [(1-\paramp)\Pi^T(\paraml,u) + \paramp B^T]_{k,j}
\end{split}
\end{align}
\[\textrm{where}\quad \Pi(\paraml,\,u) := \diag{(\textbf{A}_1,\, \textbf{A}_2,\cdots,\,,\textbf{A}_N)}\]
      Each $\textbf{A}_j$ is the $m \times m$ matrix
       \[[\textbf{A}_j]_{(k,l)}=\begin{cases}\frac{\utility_{{(l,k)}_j}\,g(u)}{z_{{(k,k)}_j}}& k\neq l\\
      \frac{\utility_{{(k,k)}_j}\,(1-g(u))}{z_{{(k,k)}_j}}& k=l
      \end{cases}\]
\[z_{{(k,k)}_j}=\utility_{(k,k)_j}+(1-2\,\utility_{(k,k)_j})\,g(\control)\]
 Here the agent's subjective utility of action $\act_i$ given state $\CE_j$ is defined as 
\begin{equation}
\label{subutil}
    \utility_{i,j}(\paraml) = \frac{\zeta(\act_i, \CE_j)^{\paraml}}{\sum_i \zeta(\act_i, \CE_j)^{\paraml}}
\end{equation} 
where $\zeta(\act_i, \CE_j)$ is the objective utility \eqref{zeta}.
\[    B:=  {\begin{bmatrix}
        \posterior(\mathcal{E}_1) & \posterior(\mathcal{E}_2)&\cdots&\posterior(\mathcal{E}_{n}) 
      \end{bmatrix} }\,\otimes\,\mathbf{1}_{m\times1}\otimes \mathbf{I}_{m\times m}
 \] Recall $\posterior$ is the posterior belief of the underlying state.
\end{enumerate}

{\bf Discussion of \eqref{Hamiltonian}, \eqref{cog_mat}, \eqref{subutil}.}

In \eqref{subutil}, $\paraml \in [0,\infty)$ models the agent's ability to discriminate between the profitability of different actions, in terms of the objective utility \eqref{zeta}. As seen in \eqref{subutil}, when $\paraml = 0$ all actions $a_j \in \mathcal{A}$ have the same probability of being chosen. When $\paraml \rightarrow \infty$, only the dominant action is chosen.

In \eqref{cog_mat}, $\paramp \in (0,1)$ interpolates between the $\Pi$ and $B$ matrices and therefore models how much the human distinguishes between the underlying states $\{\mathcal{E}_1, \dots, \mathcal{E}_{n}\}$. In our framework, $\paramp$ is a function of the control input and is given by $\paramp=u^{2l},\,l=2,3,\cdots$ to satisfy the input constraints in Sec.~\ref{machine_control}.
\label{Lindblad_param}
     $\Pi(\paraml,u)$, in \eqref{cog_mat}, acts as a transition probability matrix between nodes within individual connected components (each state grouping) of the State-Action network, as shown as the blue and red sections of the two-state example network Fig.~\ref{fig:Network}.
 The matrix $B$ is also a stochastic matrix and represents the transitions between actions.
Without loss of generality, we can assume that an external sensor makes measurements of the state, computes the Bayesian posterior, and provides this to the human.

The cognitive matrix $C(\paraml,\paramp,u)$ in~\eqref{cog_mat} and  Hamiltonian $H^u$ \eqref{Hamiltonian} act as transition matrices on the State-Action network. Taken together in the structure of (\ref{Lindblad}), these induce a quantum random walk on this network which represents the psychological evolution of state-action preferences. 
 There can be a wide variety of observed preference evolution within the two dimensional parameter space defined by ($\parama, \paraml$). We can exploit the results from \cite{6389714} to  estimate these parameters.
 \section{Additional Proofs}
 \subsection{Proof of Theorem~\ref{LindbladKraussProof}}
 \label{proof_lind}
 \begin{proof}
A first order Euler discretization of \eqref{Lindblad} yields:
\begin{align}
\begin{split}
\label{euler_discr}
\rho(t+\Delta t)&=[\mathbf{1}-T\,\Delta t\,(i\,(1-\parama)\,H^u\\&+\parama\,\frac{1}{2}\sum_{m,n}\,\gamma_{m,n}^u\,L_{(m,n)}^{\dagger}\,L_{(m,n)})]\,\rho(t) \\ &[\mathbf{1}-T\,\Delta t\,(-i\,(1-\parama)\,H^u\\ &+\parama\,\frac{1}{2}\sum_{m,n}\,\gamma_{m,n}^u\,L_{(m,n)}^{\dagger}\,L_{(m,n)})]\\ + &\parama\,\Delta t\,\sum_{m,n}\gamma_{m,n}^u L_{(m,n)}\,\rho(t)\,L_{(m,n)}^{\dagger}
\end{split}
\end{align}
where $\Delta t$ is the interval of discretization and is sufficiently small so that the bracketed terms of \eqref{euler_discr} remain positive. Denoting $\rho_{k+1}$ as $\rho(t+\Delta t)$ and $\rho_{k+T}$ as $\rho(t+\Delta t\,T)$:
\begin{equation}
\label{lindkraus}
    \rho_{k+T}=\sum_{\nu}K_{\nu,T}^{\control}\,\rho_{k}\,K_{\nu,T}^{\control \dagger}
\end{equation}
where 
\begin{align*}
\begin{split}
   K_{0,T}^{\control}&=[\mathbf{1}-T\,\Delta t\,(i\,(1-\parama)\,H^{\control}\\             &+\parama\,\frac{1}{2}\sum_{m,n}\,\gamma_{m,n}^{\control}\,L_{(m,n)}^{\dagger}\,L_{(m,n)})]   
\end{split}
\end{align*}
\[K_{\nu\neq0,T}^{\control}=\sqrt{T\,\Delta t\,\,\parama\,\gamma_{m,n}^{\control}}L_{(m,n)}, \>T\in\BN\] 
Eq.~\eqref{lindkraus} represents the discrete time psychological preference evolution. This representation is used to obtain the form (\ref{M_update}), which allows us to prove convergence of the psychological state in Sec. \ref{sec3}. 

When the human chooses action $a$, the psychological state is projected into the respective action subspace of the Hilbert Space $a \otimes \CX$, defined in \ref{machine_control}. The psychological state  evolves  as:
\begin{equation*}
    \rho_{k+T} = \frac{P_{a_k}(\sum_{\nu}K_{\nu,T}\,\rho_{k}\,K_{\nu,T}^{\dagger})P_{a_k}^{\dagger}}{\Tr(P_{a_k}(\sum_{\nu}K_{\nu,T}\,\rho_{k}\,K_{\nu,T}^{\dagger}) P_{a_k}^{\dagger})}
\end{equation*}
where $\act_k$ is a random variable taking values in $\{1,\dots,m\}$, with probability given by 
\begin{equation*}
    \PR(a_k) = \Tr(P_{a_k}(\sum_{\nu}K_{\nu,T}\,\rho_{k}\,K_{\nu,T}^{\dagger}) P_{a_k}^{\dagger})
\end{equation*}
from \eqref{eq:ProbEvo}.
$P_{a_k}$ represents the projection into the subspace corresponding to the action taken, i.e. the action acts as a quantum measurement on the psychological state. 
Now letting 
\begin{equation}
\label{M_def}
    M_{a_k}^{\control} = P_{a_k}(\sum_{\nu}K_{\nu,T}^{\control})
\end{equation}
we obtain (\ref{M_update}).
\end{proof}
\subsection{Convergence of constructed subsequences implies convergence of sequence}
\label{pf:subseq_conv}
From (\ref{limitConv}) we have $\lim_{n \to \infty}\varphi(X_n^{(k)}(\omega)) = 0 \ \forall\, \omega \in  \bar{\Omega}$. Let $\omega \in \bar{\Omega}$ and $\varphi_n^{(k)}$ denote $\varphi(X_n^{(k)}(\omega))$ and $\varphi_n$ denote $\varphi(X_n(\omega))$. \\ 
We have: $\forall\, \epsilon > 0 \ \exists\, N_k$ such that $\varphi_n^{(k)} < \epsilon \ \forall n > N_k$. Take $N^{*} = \max_{k \in \{0,\dots,T-1\}}N_k$. Suppose $\lim_{n \to \infty}\varphi_n \neq 0$: Then $\exists\, \epsilon > 0$ such that $\forall \ N \in \mathbb{N} \ \exists \ n_0 > N$ with $\varphi_{n_0} > \epsilon.$ Since the subsequences are exhaustive, i.e. for any $\varphi_n \ \exists\, k, m,$ such that $\varphi_n = \varphi_m^{(k)}$, we know that for any $\epsilon > 0$, $\varphi_n < \epsilon \ \forall\, n > N^{*}$ so such a $n_0$ does not exists for $N^{*}$ and thus by contradiction we have $\lim_{n \to \infty}\varphi_n = 0.$ \\

\subsection{Construction of $\{\sigma_r\}_{r=1}^d$}
This subsection explains how we choose the parameters $\{\sigma_r\}_{r=1}^d$ for the \ly function $V_\epsilon(\cdot)$  in \eqref{Ly_fcn}. The following two Lemmas are provided in \cite{AMINI20132683};  for completeness we list them here. \cite{AMINI20132683} uses these Lemmas to show that the limit set of the process of interest contains only the target state. We provide  intuition for this idea in the proof of Lemma~\ref{limset_restr} (Section~\ref{Sec:Qly}).
\label{ap:sigma}
\begin{lemma}[\cite{AMINI20132683}] 
\label{ap:lem:sig1} 
Consider the $d \times d$ matrix $R$ defined by 
\begin{equation*}
    \begin{split}
        R_{n_1,n_2} &= \sum_{\act_k}\, 2 \left|\bra{n_1} \frac{d\,M_{\act_k}^{u\dagger}}{d\,u}|_{u=0}\ket{n_2} \right|^2  + \\ &2\delta_{n_1,n_2}\operatorname{Re}\left( c_{\mu,n_1}\bra{n_1} \frac{d^2M_{\mu}^{u\dagger}}{du^2}|_{u=0}\ket{n_2}\right)
    \end{split}
\end{equation*}
 where $\operatorname{Re}(x)$ denotes the real part of complex valued scalar $x$, and $\delta$ denotes the Kronecker Delta. When $R$ is a nonzero matrix, the non-negative matrix $P = \id - \frac{R}{\Tr(R)}$ is a stochastic matrix (each row sums to one).
\end{lemma}
\begin{lemma}[\cite{AMINI20132683}]
\label{ap:lem:sig2}
 Assume that the directed graph $G$ of the matrix $R$ defined in Lemma~\ref{ap:lem:sig1} is strongly connected. For any $\bar{n} \in \{1,\dots,d\}\,$ there exist $d-1$ strictly positive real numbers $e_n, n \in \{1,\dots, d\} \backslash \{\bar{n}\}$, such that 
\begin{itemize}
    \item for any real numbers $\{\lambda_n\}_{n \in \{1,\dots,d\} \backslash \{\bar{n}\}}$, there exists a unique $\ \boldsymbol\sigma = \{\sigma_n\}_{n \in \{1,\dots,d\}} \in \reals^d$ with $\sigma_{\bar{n}}=0$ such that $R\,\boldsymbol\sigma = \boldsymbol\lambda$ and $\lambda_{\bar{n}} = -\sum_{n\neq \bar{n}} e_n\, \lambda_n$
    \begin{itemize}
        \item if additionally $\lambda_n < 0$ for all $n \in \{1,\dots,d\} \backslash \{\bar{n}\}$, then $\sigma_n > 0 \ \forall\, n \in \{1,\dots,d\} \backslash \{\bar{n}\}$.
    \end{itemize}
    \item for any solution $\boldsymbol\sigma \in \reals^d$  of $R\,\boldsymbol\sigma = \boldsymbol\lambda \in \reals^d$, the function $V_0 (\rho) = \sum_{n=1}^d \sigma_n \bra{n} \rho \ket{n}$ satisfies
    \begin{align*}
        \begin{split}
          &\frac{d^2\,V_0(\sum_{a_k}M_{a_k}^{u}(\ket{n}\bra{n})M_{a_k}^{u\,\dagger})}{du^2} |_{u=0} = \lambda_n\\ 
          &\forall\, n \in \{1,\dots,d\}
        \end{split}
    \end{align*}
\end{itemize}
\end{lemma}

\subsection{Martingale Proof for Density Operator Evolution}
\label{ap:martingale}
As stated in Sec.~\ref{Sec:Qly}, we prove  that under the evolution of (\ref{M_update}), $\bra{b_r} \rho_k \ket{b_r}$ is a martingale, i.e.
\[\EX[\bra{b_r} \rho_{k+T} \ket{b_r} | \mathcal{B}_k] = \bra{b_r} \rho_k \ket{b_r} \ \forall T \in \mathbb{N}\]
where the filtration $\mathcal{B}_k = \sigma(\bra{b_r} \rho_1 \ket{b_r},\dots,\bra{b_r} \rho_k \ket{b_r})$. \\
Denote
\[\mathbb{M}_{\act_k}\rho_k = \mathbb{M}_{\act_k, T=1}^0 \rho_k=\frac{M_{\act_k, 1}^0\, \rho_k\, M_{\act_k, 1}^{0\dagger}}{\Tr(M_{\act_k, 1}^0\, \rho_k\, M_{\act_k, 1}^{0\dagger})} \]

Since $\bra{b_r} \rho_k \ket{b_r} = \Tr(\ket{b_r}\bra{b_r} \rho_k)$, 
\begin{equation*}
\begin{split}
    &\EX[\Tr(\ket{b_r}\bra{b_r}\, \rho_{k+1}) | \rho_k,\control]\\ &=\sum_{\act_k = 1}^m  \Tr(M_{\act_k}\, \rho_k \, M_{\act_k}^{\dagger}) \Tr(\ket{b_r}\bra{b_r}\, \mathbb{M}_{\act_k}\, \rho_k)\\
    &=\sum_{\act_k = 1}^m \Tr(\ket{b_r}\bra{b_r}\, M_{\act_k}\, \rho_k\, M_{\act_k}^{\dagger})\\&=\Tr(\sum_{\act_k = 1}^m  M_{\act_k} M_{\act_k}^{\dagger} \ket{b_r}\bra{b_r} \rho_k)=\Tr(\ket{b_r}\bra{b_r} \rho_k).
\end{split}
\end{equation*} 
Thus $\bra{b_r} \rho_{k} \ket{b_r}$ is a martingale.

%



\bibliographystyle{IEEEtran}  
\bibliography{Bibliography}

\begin{thebibliography}{10}
\providecommand{\url}[1]{#1}
\csname url@samestyle\endcsname
\providecommand{\newblock}{\relax}
\providecommand{\bibinfo}[2]{#2}
\providecommand{\BIBentrySTDinterwordspacing}{\spaceskip=0pt\relax}
\providecommand{\BIBentryALTinterwordstretchfactor}{4}
\providecommand{\BIBentryALTinterwordspacing}{\spaceskip=\fontdimen2\font plus
\BIBentryALTinterwordstretchfactor\fontdimen3\font minus
  \fontdimen4\font\relax}
\providecommand{\BIBforeignlanguage}[2]{{%
\expandafter\ifx\csname l@#1\endcsname\relax
\typeout{** WARNING: IEEEtran.bst: No hyphenation pattern has been}%
\typeout{** loaded for the language `#1'. Using the pattern for}%
\typeout{** the default language instead.}%
\else
\language=\csname l@#1\endcsname
\fi
#2}}
\providecommand{\BIBdecl}{\relax}
\BIBdecl

\bibitem{martinez2016quantum}
I.~Mart{\'\i}nez-Mart{\'\i}nez and E.~S{\'a}nchez-Burillo, ``Quantum stochastic
  walks on networks for decision-making,'' \emph{Scientific reports}, vol.~6,
  no.~1, pp. 1--13, 2016.

\bibitem{kvam2021temporal}
P.~D. Kvam, J.~R. Busemeyer, and T.~J. Pleskac, ``Temporal oscillations in
  preference strength provide evidence for an open system model of constructed
  preference,'' \emph{Scientific reports}, vol.~11, no.~1, pp. 1--15, 2021.

\bibitem{busemeyer2020application}
J.~Busemeyer, Q.~Zhang, S.~Balakrishnan, and Z.~Wang, ``Application of
  quantum—{M}arkov open system models to human cognition and decision,''
  \emph{Entropy}, vol.~22, no.~9, p. 990, 2020.

\bibitem{askarpour2019formal}
M.~Askarpour, D.~Mandrioli, M.~Rossi, and F.~Vicentini, ``Formal model of human
  erroneous behavior for safety analysis in collaborative robotics,''
  \emph{Robotics and computer-integrated Manufacturing}, vol.~57, pp. 465--476,
  2019.

\bibitem{belanche2020consumer}
D.~Belanche, C.~Flavi{\'a}n, and A.~P{\'e}rez-Rueda, ``Consumer empowerment in
  interactive advertising and ewom consequences: The pitre model,''
  \emph{Journal of Marketing Communications}, vol.~26, no.~1, pp. 1--20, 2020.

\bibitem{lu2012recommender}
L.~L{\"u}, M.~Medo, C.~H. Yeung, Y.-C. Zhang, Z.-K. Zhang, and T.~Zhou,
  ``Recommender systems,'' \emph{Physics reports}, vol. 519, no.~1, pp. 1--49,
  2012.

\bibitem{9280374}
A.~R. Hota and S.~Sundaram, ``Controlling human utilization of failure-prone
  systems via taxes,'' \emph{IEEE Transactions on Automatic Control}, vol.~66,
  no.~12, pp. 5772--5787, 2021.

\bibitem{HOEY2010503}
\BIBentryALTinterwordspacing
J.~Hoey, P.~Poupart, A.~von Bertoldi, T.~Craig, C.~Boutilier, and
  A.~Mihailidis, ``Automated handwashing assistance for persons with dementia
  using video and a partially observable markov decision process,''
  \emph{Computer Vision and Image Understanding}, vol. 114, no.~5, pp.
  503--519, 2010, special issue on Intelligent Vision Systems. [Online].
  Available:
  \url{https://www.sciencedirect.com/science/article/pii/S1077314210000354}
\BIBentrySTDinterwordspacing

\bibitem{AMINI20132683}
\BIBentryALTinterwordspacing
H.~Amini, R.~A. Somaraju, I.~Dotsenko, C.~Sayrin, M.~Mirrahimi, and P.~Rouchon,
  ``Feedback stabilization of discrete-time quantum systems subject to
  non-demolition measurements with imperfections and delays,''
  \emph{Automatica}, vol.~49, no.~9, pp. 2683--2692, 2013. [Online]. Available:
  \url{https://www.sciencedirect.com/science/article/pii/S0005109813003269}
\BIBentrySTDinterwordspacing

\bibitem{qin2019lyapunov}
Y.~Qin, M.~Cao, and B.~D. Anderson, ``Lyapunov criterion for stochastic systems
  and its applications in distributed computation,'' \emph{IEEE Transactions on
  Automatic Control}, vol.~65, no.~2, pp. 546--560, 2019.

\bibitem{morgenstern1953theory}
O.~Morgenstern and J.~Von~Neumann, \emph{Theory of games and economic
  behavior}.\hskip 1em plus 0.5em minus 0.4em\relax Princeton university press,
  1953.

\bibitem{savage1951theory}
L.~J. Savage, ``The theory of statistical decision,'' \emph{Journal of the
  American Statistical association}, vol.~46, no. 253, pp. 55--67, 1951.

\bibitem{kahneman1982judgment}
D.~Kahneman, S.~P. Slovic, P.~Slovic, and A.~Tversky, \emph{Judgment under
  uncertainty: Heuristics and biases}.\hskip 1em plus 0.5em minus 0.4em\relax
  Cambridge university press, 1982.

\bibitem{kahneman2013prospect}
D.~Kahneman and A.~Tversky, ``Prospect theory: An analysis of decision under
  risk,'' in \emph{Handbook of the fundamentals of financial decision making:
  Part I}.\hskip 1em plus 0.5em minus 0.4em\relax World Scientific, 2013, pp.
  99--127.

\bibitem{busemeyer2012quantum}
\BIBentryALTinterwordspacing
J.~Busemeyer and P.~Bruza, \emph{Quantum Models of Cognition and Decision},
  ser. Quantum Models of Cognition and Decision.\hskip 1em plus 0.5em minus
  0.4em\relax Cambridge University Press, 2012. [Online]. Available:
  \url{https://books.google.com/books?id=0vxvhTG\_ZLAC}
\BIBentrySTDinterwordspacing

\bibitem{khrennikov2010ubiquitous}
A.~Khrennikov, \emph{Ubiquitous quantum structure}.\hskip 1em plus 0.5em minus
  0.4em\relax Springer, 2010.

\bibitem{yukalov2010mathematical}
V.~I. Yukalov and D.~Sornette, ``Mathematical structure of quantum decision
  theory,'' \emph{Advances in Complex Systems}, vol.~13, no.~05, pp. 659--698,
  2010.

\bibitem{khrennikov2009quantum}
A.~Y. Khrennikov and E.~Haven, ``Quantum mechanics and violations of the
  sure-thing principle: The use of probability interference and other
  concepts,'' \emph{Journal of Mathematical Psychology}, vol.~53, no.~5, pp.
  378--388, 2009.

\bibitem{aerts2011quantum}
D.~Aerts, J.~Broekaert, M.~Czachor, and B.~D’Hooghe, ``A quantum-conceptual
  explanation of violations of expected utility in economics,'' in
  \emph{International Symposium on Quantum Interaction}.\hskip 1em plus 0.5em
  minus 0.4em\relax Springer, 2011, pp. 192--198.

\bibitem{trueblood2011quantum}
J.~S. Trueblood and J.~R. Busemeyer, ``A quantum probability account of order
  effects in inference,'' \emph{Cognitive science}, vol.~35, no.~8, pp.
  1518--1552, 2011.

\bibitem{busemeyer2011quantum}
J.~R. Busemeyer, E.~M. Pothos, R.~Franco, and J.~S. Trueblood, ``A quantum
  theoretical explanation for probability judgment errors.''
  \emph{Psychological review}, vol. 118, no.~2, p. 193, 2011.

\bibitem{busemeyer2009empirical}
J.~R. Busemeyer, Z.~Wang, and A.~Lambert-Mogiliansky, ``Empirical comparison of
  {M}arkov and quantum models of decision making,'' \emph{Journal of
  Mathematical Psychology}, vol.~53, no.~5, pp. 423--433, 2009.

\bibitem{asano2012quantum}
M.~Asano, I.~Basieva, A.~Khrennikov, M.~Ohya, and Y.~Tanaka, ``Quantum-like
  dynamics of decision-making,'' \emph{Physica A: Statistical Mechanics and its
  Applications}, vol. 391, no.~5, pp. 2083--2099, 2012.

\bibitem{aeyels1998new}
D.~Aeyels and J.~Peuteman, ``A new asymptotic stability criterion for nonlinear
  time-variant differential equations,'' \emph{IEEE Transactions on Automatic
  Control}, vol.~43, no.~7, pp. 968--971, 1998.

\bibitem{nedic2014distributed}
A.~Nedi{\'c} and A.~Olshevsky, ``Distributed optimization over time-varying
  directed graphs,'' \emph{IEEE Transactions on Automatic Control}, vol.~60,
  no.~3, pp. 601--615, 2014.

\bibitem{nedic2009distributed}
A.~Nedic and A.~Ozdaglar, ``Distributed subgradient methods for multi-agent
  optimization,'' \emph{IEEE Transactions on Automatic Control}, vol.~54,
  no.~1, pp. 48--61, 2009.

\bibitem{tahbaz2009consensus}
A.~Tahbaz-Salehi and A.~Jadbabaie, ``Consensus over ergodic stationary graph
  processes,'' \emph{IEEE Transactions on Automatic Control}, vol.~55, no.~1,
  pp. 225--230, 2009.

\bibitem{cao2008agreeing}
M.~Cao, A.~S. Morse, and B.~D. Anderson, ``Agreeing asynchronously,''
  \emph{IEEE Transactions on Automatic Control}, vol.~53, no.~8, pp.
  1826--1838, 2008.

\bibitem{tahbaz2008necessary}
A.~Tahbaz-Salehi and A.~Jadbabaie, ``A necessary and sufficient condition for
  consensus over random networks,'' \emph{IEEE Transactions on Automatic
  Control}, vol.~53, no.~3, pp. 791--795, 2008.

\bibitem{liu2017exponential}
J.~Liu, A.~S. Morse, A.~Nedi{\'c}, and T.~Ba{\c{s}}ar, ``Exponential
  convergence of a distributed algorithm for solving linear algebraic
  equations,'' \emph{Automatica}, vol.~83, pp. 37--46, 2017.

\bibitem{mou2015distributed}
S.~Mou, J.~Liu, and A.~S. Morse, ``A distributed algorithm for solving a linear
  algebraic equation,'' \emph{IEEE Transactions on Automatic Control}, vol.~60,
  no.~11, pp. 2863--2878, 2015.

\bibitem{pothos2009quantum}
E.~M. Pothos and J.~R. Busemeyer, ``A quantum probability explanation for
  violations of ‘rational’decision theory,'' \emph{Proceedings of the Royal
  Society B: Biological Sciences}, vol. 276, no. 1665, pp. 2171--2178, 2009.

\bibitem{Pearle_2012}
\BIBentryALTinterwordspacing
P.~Pearle, ``Simple derivation of the lindblad equation,'' \emph{European
  Journal of Physics}, vol.~33, no.~4, pp. 805--822, apr 2012. [Online].
  Available: \url{https://doi.org/10.1088/0143-0807/33/4/805}
\BIBentrySTDinterwordspacing

\bibitem{6160433}
H.~Amini, P.~Rouchon, and M.~Mirrahimi, ``Design of strict control-{L}yapunov
  functions for quantum systems with qnd measurements,'' in \emph{2011 50th
  IEEE Conference on Decision and Control and European Control Conference},
  2011, pp. 8193--8198.

\bibitem{kushner1971introduction}
\BIBentryALTinterwordspacing
H.~Kushner, \emph{Introduction to Stochastic Control}.\hskip 1em plus 0.5em
  minus 0.4em\relax Holt, Rinehart and Winston, 1971. [Online]. Available:
  \url{https://books.google.com/books?id=vUzvAAAAMAAJ}
\BIBentrySTDinterwordspacing

\bibitem{6389714}
Z.~Xue, H.~Lin, and T.~H. Lee, ``Identification of unknown parameters for a
  class of two-level quantum systems,'' \emph{IEEE Transactions on Automatic
  Control}, vol.~58, no.~7, pp. 1805--1810, 2013.

\end{thebibliography}

%

\begin{IEEEbiography}{Luke Snow}
Biography text here.
\end{IEEEbiography}

\begin{IEEEbiographynophoto}{Shashwat Jain}
Biography text here.
\end{IEEEbiographynophoto}


\begin{IEEEbiographynophoto}{Vikram Krishnamurthy}
Biography text here.
\end{IEEEbiographynophoto}




\end{document}